\documentclass[lettersize,onecolumn]{IEEEtran}
\usepackage{amsmath,amsfonts}
\usepackage{algorithmic}
\usepackage{algorithm}
\usepackage{array}
\usepackage[caption=false,font=normalsize,labelfont=sf,textfont=sf]{subfig}
\usepackage{textcomp}
\usepackage{stfloats}
\usepackage{url}
\usepackage{verbatim}
\usepackage{graphicx}
\usepackage{mathtools}
\usepackage{bm}
\usepackage{mathrsfs}
\usepackage{booktabs}
\usepackage{array}
\usepackage{tikz}
\usepackage{cite}
\newtheorem{theorem}{Theorem}
\newtheorem{lemma}[theorem]{Lemma}
\newtheorem{proposition}[theorem]{Proposition}
\newtheorem{corollary}[theorem]{Corollary}
\newtheorem{definition}[theorem]{Definition}
\newtheorem{remark}[theorem]{Remark}

\begin{document}

\title{A Complete Characterization of Tensorizable $f$-divergences}

\author{
Rodrigo Cruz,
Flavio P. Calmon 
and Qian Yu 
\thanks{Rodrigo Cruz and Flavio P. Calmon are with the School of Engineering and Applied Sciences, Harvard University, Cambridge, MA USA
(e-mail: ccruzflores@g.harvard.edu; flavio@seas.harvard.edu).}%
\thanks{Qian Yu is with the Department of Electrical and Computer Engineering, University of California at Santa Barbara, Santa Barbara, CA USA
(e-mail: qianyu02@ucsb.edu).}%
\thanks{Corresponding author: Rodrigo Cruz}%
}

\maketitle

\begin{abstract}
 Csiszar's formulation of the $f$-divergence introduced a vast family of functionals for quantifying dissimilarity between probability distributions. However, many applications in statistics and information theory rely only on a few $f$-divergences, such as the Kullback-Leibler divergence, the $\chi^2$-divergence, and the squared Hellinger distance. These divergences are especially useful because they admit simple compositional formulas under product measures, a property sometimes referred to as \emph{tensorization}. In this work, we refine a formalism of tensorization previously introduced in the literature. Then, we show that any possible tensorization formula has a multi-affine form characterized by a single parameter, and identify all tensorizable $f$-divergences under our adopted notion of tensorization.  
\end{abstract}


\section{Introduction}
Measuring discrepancy between probability distributions is central to information theory and statistics. Csisz\'ar's $f$-divergence \cite{csiszar1967information} provides a rich family of functionals for quantifying this discrepancy. Despite the myriad of $f$-divergences, a particular subset repeatedly appears in theory and practice: \emph{tensorizable} $f$-divergences (e.g. \cite{kakutani_equivalence_1948, kraft_conditions_1955, yang_information-theoretic_1999,li2016renyi, nowozin2016f,  mironov2017renyi, balle2020hypothesis}). Intuitively, a tensorizable $f$-divergence evaluated at a pair of product measures can be decomposed into an expression that depends only on marginal discrepancies. Prominent examples of tensorizable $f$-divergences are the squared Hellinger distance and the Kullback-Leibler (KL) and $\chi^2$ divergences. Recently, \cite{CruzRodrigo2025Tof} formalized the notion of a tensorizable $f$-divergence $D_f$ as one for which there exists a sequence of continuous maps $\{\tau_n\}_{n = 1}^\infty$ such that
\[
D_f(P_1 \otimes \cdots \otimes P_n \Vert Q_1 \otimes \cdots \otimes Q_n) = \tau_n(D_f(P_1 \Vert Q_1), \dots, D_f(P_n \Vert Q_n)).
\]
These formulas keep tensorizable $f$-divergences tractable even for measures with many factors. As a result, they are especially useful in applications where the operationally meaningful discrepancy measure is either difficult to compute or difficult to compose under product measures. For example, in the Neyman-Pearson formulation of binary hypothesis testing, the total variation distance determines the optimal testing error rate \cite{LeCamLucien1986AMiS}. However, computing this metric is often complicated even for tests involving independent and identically distributed observations of well-behaved distributions. Thus, we rely on Pinsker or Kraft's inequalities to bound or estimate this quantity. Similarly, in $(\epsilon, \delta)$-differential privacy, the hockey-stick divergence gives the information leakage of a privacy-preserving mechanism \cite{balle2020hypothesis}. However, when private data is queried multiple times, quantifying the cumulative leakage can be challenging --- a problem known as privacy accounting \cite{mcsherry2009privacy, abadi2016deep, zhu2022optimal}. Accordingly, one of the most influential approaches to privacy accounting \cite{mironov2017renyi} relies on monotone transformations of Hellinger divergences of order $\alpha$, which are tensorizable $f$-divergences. 

Tensorizable $f$-divergences also appear prominently in online learning and stochastic optimization. The most common application is the linearization of KL divergences, often combined with standard information-theoretic tools such as Pinsker’s inequality, to derive matching optimality bounds. Such techniques underlie many foundational results, including instance-dependent and minimax regret bounds in multi-armed bandits \cite{10.1016/0196-8858, pmlr-v19-garivier11a, 10.5555/795662.796294, audibert:hal-00834882}, as well as minimax sample complexity characterizations in convex and non-convex stochastic optimization \cite{7055287, pmlr-v206-yu23a, NEURIPS2024_b444ad72}. More recently, a growing line of work has shown that approaches beyond local linearization can play a fundamental role in establishing sharp information-theoretic lower bounds and optimality guarantees in interactive decision-making and adaptive stochastic optimization problems \cite{pmlr-v195-foster23b, NEURIPS2024_8a23a95e}. 

\subsection{Contributions and Methods}

Despite their ubiquity, tensorizable $f$-divergences have not yet been formally characterized. We address this gap by building upon the notion of tensorization defined in \cite{CruzRodrigo2025Tof} to show that:
\begin{center}
\begin{tikzpicture}
\node[
  draw,
  rounded corners=3pt,
  line width=0.6pt,
  inner xsep=10pt,
  inner ysep=6pt,
  text width=0.85\linewidth
] {The only tensorizable $f$-divergences are either linear combinations (with nonnegative coefficients) of the KL and reverse KL divergences or power divergences (see Table~\ref{tab:f-divergences} for a list of common $f$-divergences). };
\end{tikzpicture}
\end{center}

To achieve this, we present three technical contributions: 
\begin{enumerate}
    \item We show that every admissible tensorization formula belongs to a one-parameter family of ``multi-affine'' functions. We call the parameter indexing this family the tensorization parameter. Our argument exploits how the Lebesgue integral behaves linearly with respect to convex combinations of measures. This forces every $f$-divergence to be affine when viewed as a function of log-likelihood ratio (LLR) distributions. Tensorization extends this affine structure, making the $f$-divergence on a pair of product measures depend linearly on each marginal discrepancy.

    \item We introduce the subclass of symmetrized tensorizable $f$-divergences and prove that each of its elements is uniquely characterized by two real numbers. Specifically, when restricted to two-point symmetric LLR distributions, a tensorizable $f$-divergence induces a quadratic functional equation whose solutions are characterized by the tensorization parameter and the quadratic rate at which the divergence vanishes as the support of the two-point LLR distribution converges to $\{0\}$. We call this quadratic rate the local characteristic of the $f$-divergence.

    \item We identify every convex function $f$ that induces a tensorizable $f$-divergence $D_f$. Using the fact that the tensorization parameter and local characteristic determine the form of a symmetrized tensorizable $f$-divergence, we derive a linear equation that any convex function $f$ that induces $D_f$ must satisfy. Intuitively, once this form is fixed, the resulting solution space has two degrees of freedom: one corresponds to the freedom to add an affine term to $f$ without affecting $D_f$, and the other corresponds to the fact that $f$ and $x \mapsto x f(1/x)$ induce the same symmetrized $f$-divergence. Once these two aspects are fixed, we can uniquely identify $f$. Formally, this argument reduces to identifying a two-element basis for the solution space of the associated homogeneous linear equation.
\end{enumerate}

\subsection{Related Work}
\paragraph{Tensorization of $f$-divergences} Our results complete \cite{CruzRodrigo2025Tof} efforts to characterize tensorizable $f$-divergences. Specifically, \cite{CruzRodrigo2025Tof} showed, under mild regularity assumptions on $f$, that if the tensorization rule for $D_f$ is a multivariate polynomial function of degree at most $n$, then $D_f$ must be a positive multiple of either KL, reverse KL or a Hellinger divergence of order $\alpha$. In comparison, we first identify every admissible tensorization rule, which allows us to characterize all tensorizable $f$-divergence without imposing any assumptions on the tensorization formulas. 

\paragraph{Axiomatic Characterizations of Information Measures}
The foundational work of Shannon \cite{shannon1948mathematical} spawned a prolific area of research focused on characterizing  information measures through axioms, especially compositional axioms that lead to functional equations. Hobson \cite{HobsonArthur1969Anto} provided an initial axiomatic characterization of the KL divergence by identifying all functionals on pairs of discrete probability measures that satisfy continuity, permutation invariance, recursivity, monotonicity, and that vanish at identical pairs. Related results were later obtained by Kannappan and Ng through recursivity \cite{KannappanPl1973Msof} and through additivity under product measures \cite{kannappan1974functional}. Csiszar provides a broad review of this line of work in \cite{csiszar2008axiomatic}. Furthermore, a comprehensive treatment of functional equations arising in information theory can be found in \cite{ebanks1998characterization}; as noted in \cite{csiszar2008axiomatic}, some of the results in \cite{ebanks1998characterization} imply characterizations of the KL and power divergences. Our work is closely related to these axiomatic characterizations, but differs in that we do not impose a specific compositional rule a priori. Instead, we first determine which compositional rules are compatible with the notion of tensorization and then characterize the $f$-divergences that admit those rules. 

\paragraph{Characterizations of Tsallis relative entropy} 
The Hellinger divergence of order $\alpha$ (defined in Table~\ref{tab:f-divergences}) naturally extends Tsallis $\alpha$-entropy \cite{TsallisConstantino2009ItNS}, which is why $H_\alpha$ is also commonly known as the Tsallis $\alpha$-relative entropy.\footnote{The standard notation for the Tsallis entropy uses the symbol $q$ instead of $\alpha$. However, this convention would inevitably clash with our repeated use of the symbol $Q$ to denote a probability measure.}  Notably, $H_\alpha$ decomposes as an $(\alpha-1)$-sum\footnote{We define the $\gamma$-sum of two real numbers as $a \oplus_\gamma b \coloneqq a + b + \gamma ab$, for real numbers $a,b, \gamma \in \mathbb{R}$} under product measures:
\begin{equation}
    \label{eq:relative-q-sum}
    H_\alpha(P_1 \otimes P_2 \Vert Q_1 \otimes Q_2) = H_\alpha(P_1 \Vert Q_1) + H_\alpha(P_2 \Vert Q_2) + (\alpha - 1)H_\alpha(P_1 \Vert Q_1)H_\alpha(P_2 \Vert Q_2).
\end{equation} 
This identity has motivated multiple axiomatic characterizations of the Tsallis relative entropy based on slight variations of the identity in \eqref{eq:relative-q-sum} (see, for example, \cite{suyari2004generalization, furuichi2005uniqueness, leinster2019short}). 
To the best of our knowledge, all such axiomatic characterizations of $H_\alpha$ use stronger versions of \eqref{eq:relative-q-sum} that explicitly introduce the powers $P^\alpha$ and $Q^{1-\alpha}$ into the picture. Such additional structure is warranted because, as \cite{jizba2017uniqueness} argue, the $(\alpha-1)$-sum property in \eqref{eq:relative-q-sum} is not unique to the functionals $\{H_\alpha\}_{\alpha \in \mathbb{R}_+\setminus\{1\}}$. Our work proves that, in this regime, tensorization is enough to uniquely identify Tsallis $\alpha$-relative entropy up to scaling for $\alpha \in \mathbb{R}_+\setminus\{1\}$.

\paragraph{Characterization of additive functionals} The
$f$-divergence belongs to a broader class of functionals called divergences, which are only required to be non-negative and vanish on pairs of identical measures. Within this broader class, \cite{MuXiaosheng2021FBDi} prove that any additive divergence satisfying the Data Processing Inequality is an integral of Rényi divergences. Similarly, \cite{mu2024monotone} characterize additive statistics as integrals of the normalized cumulant generating function. These works share with ours the same organizing principle: characterizing functionals of probability measures by how they behave under independence. Our setting differs in that we study a more general notion of composition --- namely, tensorization --- while we restrict our attention to the smaller class of $f$-divergences.

\section{Preliminaries}
For two probability measures $P$ and $Q$, we write $P \otimes Q$ to denote their product measure, which is uniquely defined on the product $\sigma$-algebra.
When $P$ and $Q$ are defined on the same measurable space, we say that $P$ is absolutely continuous with respect to $Q$ if $Q(A)=0$ implies $P(A)=0$ for every measurable set $A$; equivalently, $Q$ dominates $P$. Throughout this work $f$ exclusively denotes a real convex function with domain $(0, \infty)$ satisfying $f(1) = 0$. We denote the limits $\lim_{x \to 0}f(x)$ and $\lim_{x \to \infty} \frac{f(x)}{x}$ as $f(0)$ and $f'(\infty)$, respectively. Note that these limits always exist in $(- \infty, \infty]$ due to convexity.

\begin{definition}[$f$-divergence]
\label{def:f-divergence}
    Let $f:(0,\infty) \to \mathbb{R}$ be a convex function with $f(1) = 0$. For two probability measures $P$ and $Q$ on the same measurable space, let $p$ and $q$ be their corresponding Radon–Nikodym derivatives with respect to a common dominating $\sigma$-finite measure $\lambda$. 
    The \emph{$f$-divergence} is defined as
    \begin{equation}
    \label{eq:definition_Df}
    D_f(P \Vert Q) \coloneqq \int_{p>0,q>0} q f\left(\frac{p}{q}\right) \,d\lambda + f(0)Q(\{p = 0\}) + f^{\prime}(\infty)P(\{q = 0\}). 
    \end{equation}
    with the convention that if $P(\{q = 0\}) = 0$, then we take the product $f^{\prime}(\infty)\cdot P(\{q = 0\})$ to be zero.   
\end{definition}
\begin{remark}
    It is a well-known fact that $D_f(P \Vert Q)$ is invariant to the choice of dominating measure $\lambda$, as stated in \cite[Remark 7.2]{polyanskiy2025information}. 
\end{remark}

\begin{remark}
    It is crucial to note that different convex functions can give rise to the same $f$-divergence. For example, if we define $g(t) = f(t) + B(t-1)$, for some $B \in \mathbb{R}$, then $D_f(P \Vert Q) = D_g(P \Vert Q)$ for every pair of probability measures $P$ and $Q$. Additionally, the value $D_f(P\Vert Q)$ alone does not admit a universal interpretation of how close P and Q are. Rather, such values are most informative within a fixed divergence, for example when comparing pairs of measures or studying convergence. For this two reasons, the same $f$-divergence may appear in the literature under different normalizations, or even under different names. We refer the reader to Table~\ref{tab:f-divergences} for the exact definitions used herein.
\end{remark}

Given an $f$-divergence $D_f$, we denote its finite-value range by
\begin{equation}
    \label{eq:def_range}
    \mathcal{R}(D_f) 
    \coloneqq \{D_f(P \Vert Q) : \text{$P,Q$ are probability measures with } D_f(P \Vert Q) < \infty\}.
\end{equation}
In light of \cite[Theorem 2]{vajda1972f}, if $f(0)$ and $f'(\infty)$ are finite, then $\mathcal{R}(D_f) = [0, f(0) + f'(\infty)]$. Otherwise, we have that $\mathcal{R}(D_f) = [0, \infty)$.

\begin{table}[t]
    \centering
    \small
    \setlength{\tabcolsep}{6pt}
    \renewcommand{\arraystretch}{1.15}
    \begin{tabular}{@{} p{0.58\linewidth} >{\centering\arraybackslash}p{0.18\linewidth} c @{}}
        \toprule
        \textbf{Name} & \textbf{Notation} & \textbf{Representative $f$} \\
        \midrule
        Kullback-Leibler (KL) divergence
        & $D_\mathrm{KL}(P\Vert Q)$
        & $x\log x$ \\
        Reverse KL divergence
        & $D_\mathrm{KL}(Q\Vert P)$
        & $ -\log x$ \\
        Hellinger divergence of order $\alpha$ ($\alpha\in\mathbb{R}_+\setminus\{1\}$)
        & $H_\alpha(P\Vert Q)$
        & $\frac{x^\alpha-1}{\alpha-1}$ \\
        Power divergence ($\alpha \in\mathbb{R}\setminus\{0,1\}, c \in \mathbb{R}, c\alpha(\alpha-1)\geq 0$) & $D_{\alpha,c}(P\Vert Q)$ & $c(x^\alpha-1)$ \\
        Total variation distance
        & $\mathrm{TV}(P,Q)$
        & $\frac{1}{2}\lvert x-1\rvert$ \\
        \bottomrule
    \end{tabular}
    \caption{Prominent $f$-divergences, the notation we adopt for them and one representative of their associated convex functions.}
    \label{tab:f-divergences}
\end{table}

Next, we present a refinement of the definition of tensorization introduced in \cite{CruzRodrigo2025Tof}.

\begin{definition}[Tensorization]
\label{def:tensorization_base}
    Let $D_f$ be an $f$-divergence and let $(E, \mathcal{E})$ and  $(F, \mathcal{F})$ be arbitrary measurable spaces. We say that $D_f$ tensorizes if there exists a function $\tau_f: \mathcal{R}(D_f)^{2}\to\mathcal{R}(D_f)$ such that 
    \begin{equation}
        \label{eq:tensorization_base}
        D_f(P_1 \otimes P_2 \Vert Q_1 \otimes Q_2) = \tau_f\bigl(D_f(P_1 \Vert Q_1), D_f(P_2 \Vert Q_2)\bigr), 
    \end{equation}
    for any probability measures $P_1, Q_1$ on $(E, \mathcal{E})$ and $P_2, Q_2$ on $(F, \mathcal{F})$ with \[D_f(P_1 \Vert Q_1), D_f(P_2 \Vert Q_2) < \infty.\]
    Furthermore, we refer to $\tau_f$ as the tensorization base of $D_f$. 
\end{definition}

\begin{remark}      
   Our analysis mainly focuses on products of just two probability measures as we can use an inductive argument to extend any tensorization base to account for $n \geq 2$ factors. 
\end{remark}

\begin{remark}
We require tensorization to hold for probability measures on arbitrary measurable spaces, which is the most natural formalization. However, all our results extend to the setting where measures are confined to any subclass of measurable spaces closed under products; i.e., containing the product space of any two of its elements. See Appendix~\ref{appendix:domain-Df} for details.
\end{remark}

Our next lemma presents three key algebraic properties of tensorization bases. These properties are a reflection of how Radon-Nikodym derivatives combine under product measures. This result originally appeared in \cite[Lemma 2]{CruzRodrigo2025Tof}, but it still holds true under our refined definition of tensorization.
\begin{lemma}
\label{lem:tensorization-base-properties}
Let $\tau_f$ be a tensorization base. Then for any $d_1,d_2,d_3 \in \mathcal{R}(D_f)$, $\tau_f$ satisfies:
\begin{enumerate}
    \item \emph{Symmetry:} \[\tau_f(d_{1},d_{2}) = \tau_f(d_{2},d_{1}).
    \]
    \item \emph{Associativity:} \[\tau_f(\tau_f(d_1,d_2),d_3)=\tau_f(d_1,\tau_f(d_2,d_3)).
    \]
    \item \emph{Marginalization:} 
    \[\tau_f(0,d_1) = d_1 =\tau_f(d_1,0),\]
    where $0 \in \mathcal{R}(D_f)$ for every $D_f$.
\end{enumerate}
\end{lemma}

\section{Characterization of Tensorization Bases}
\label{sec:tensorization-bases-are-affine}
In this section, we identify every possible tensorization base function by interpreting $D_f(P \Vert Q)$ as a function of the log-likelihood ratio (LLR) distribution of $Q$ against $P$. Under this representation, $D_f$ behaves affinely with respect to mixtures of LLR distributions. Tensorization then propagates this affine structure to the divergence values, forcing every section of $\tau_f$ to be an affine map (Lemma~\ref{lem:base-is-affine-fn}). Finally, the convention that $f$-divergence preserves the indiscernibility of identicals and the marginalization property of tensorization bases determine the explicit formula of $\tau_f$ (Theorem~\ref{thm:functional_form_base}).

For completeness, we include a formal definition of the LLR distribution for arbitrary pairs of probability measures. Historically, the likelihood ratio is often defined using probability density or probability mass functions, but it is standard to extend this notion using the Radon-Nikodym derivative (see, for example, \cite{HalmosPaulR.1949AotR, LindleyD.V.1953SI,  yu2023ising, gonccalves2024likelihood}).

\begin{definition}[LLR Distribution]
\label{def:LLR-distribution-RN}
    Given a measurable function $\phi: (E, \mathcal{E}) \to (F, \mathcal{F})$ and a probability measure $P$ on $(E, \mathcal{E})$, we define the push-forward measure
    \begin{equation}
        \label{eq:def-push-forward}
        P \circ \phi^{-1}(A) \coloneqq P(\phi^{-1}(A)), \quad A \in \mathcal{F}.
    \end{equation}
    Let $Q$ be another probability measure on $(E, \mathcal{E})$. Take $p \coloneqq \frac{dP}{d\lambda}$ and $q \coloneqq \frac{dQ}{d\lambda}$ as the corresponding Radon–Nikodym derivatives with respect to a common dominating $\sigma$-finite measure $\lambda$. We call $\frac{p}{q}$ a likelihood ratio and define the LLR distribution of $Q$ against $P$ as  
    \begin{equation}
        \label{eq:def-LLR-pre-image}
        \hat{\mu}(P \Vert Q) \coloneqq Q \circ \left( \log \frac{q}{p}\right)^{-1},
        \end{equation}
    under the convention that $\log \frac{1}{0} \coloneqq +\infty$. 
    Note that $\hat{\mu}$ is invariant to the choice of $\lambda$ (see Lemma~\ref{lem:invariance_LLR_as}).
\end{definition}
\begin{remark}
    Equation \eqref{eq:def-LLR-pre-image} gives a probability measure on the extended real line such that
    \[
    \hat{\mu}(P\Vert Q)(\{-\infty\}) = 0.
    \] 
    However, to avoid technical issues when adding log-likelihood ratios, we will identify $\hat{\mu}(P \Vert Q)$ with its restriction to $\mathbb{R}_{+\infty} \coloneqq \mathbb{R}\cup\{+\infty\}$. 
\end{remark}

Naturally, the value $D_f(P \Vert Q)$ is determined by the LLR distribution of $Q$ against $P$. The proof of this statement is a straightforward application on an elementary fact from measure theory: $\int g \circ h \, d\mu = \int g \, d(\mu \circ h^{-1})$ \cite[Lemma 1.24]{kallenberg1997foundations}. 
\begin{lemma}
    Let $D_f$ be an $f$-divergence. Then, for probability measures $P$ and $Q$ on the same measurable space, 
    \begin{equation}
        \label{eq:divergence-LLR-relation}
        D_f(P \Vert Q) = \int (f \circ \exp^{-})\, d\hat{\mu}(P \Vert Q) + f^{\prime}(\infty)\left(1 - \int\exp^{-} \,d\hat{\mu}(P \Vert Q)\right),
    \end{equation}
    where $\exp^{-}: r \mapsto \exp(-r)$. 
\end{lemma}

It is well-known that an arbitrary measure $\mu$ on $(-\infty, \infty]$ is an LLR distribution if and only if $\int \exp^-\,d\mu \leq 1$, as stated in \cite{yu2023ising} (see also Lemma~\ref{lem:LLR-equivalent-defs} in the Appendix). Thus, we abuse notation slightly and define
\begin{equation}
    \label{eq:def_f-divergence_LLR}
    D_f(\mu) \coloneqq \int (f \circ \exp^{-})\, d\mu + f^{\prime}(\infty)\left(1 - \int\exp^{-} \,d\mu\right),
\end{equation}
for any probability measure $\mu$ on $(-\infty, \infty]$ with $\int \exp^-\,d\mu \leq 1$. 

Crucially, under product measures, the LLR is given by the sum of marginal log-likelihood ratios. Thus, $\hat{\mu}(P_1 \otimes P_2 \Vert Q_1 \otimes Q_2)$ is the convolution of $\hat{\mu}(P_1 \Vert Q_1)$ and $\hat{\mu}(P_2 \Vert Q_2)$. More importantly, we have deduced that tensorization is directly connected to how $D_f$, understood as the map in \eqref{eq:def_f-divergence_LLR}, behaves with respect to the  convolution of measures. With this in mind, we establish some notation.

\begin{definition}
\label{def:convolution}
For two probability measures $\mu_1$ and $\mu_2$ on $(\mathbb{R}_{+\infty}, \mathcal{B}(\mathbb{R}_{+\infty}))$, we define the convolution $\mu_1*\mu_2$ as
\begin{equation}
    \mu_1*\mu_2(A) \coloneqq \mu_1 \otimes \mu_2\left(\left\{(x,y) \in \mathbb{R}_{+\infty}^2 : x + y \in A\right\}\right), \quad \forall A \in \mathcal{B}(\mathbb{R}_{+\infty}).
\end{equation}
\end{definition}

Our next lemma presents the correspondence between pairs of product measures and the convolution of the marginal LLR distributions in terms of $D_f$.
\begin{lemma}
\label{lemma:Df-df}
    Let $P_1, Q_1$ and $P_2, Q_2$ be pairs of probability measures on the same measurable space, respectively. Denote $\mu_1 \coloneqq \hat{\mu}(P_1 \Vert Q_1)$ and $\mu_2 \coloneqq \hat{\mu}(P_2 \Vert Q_2)$. The convolution $\mu_1*\mu_2$ is an LLR distribution, and
    \begin{equation}
        \label{eq:Df-df}
        D_f(\mu_1 * \mu_2) = D_f(P_1 \otimes P_2 \Vert Q_1 \otimes Q_2).
    \end{equation}
\end{lemma}

We now present a key structural fact for our characterization: the map $\mu_1 \mapsto D_f(\mu_1 * \mu_2)$ is affine under mixtures of LLR distributions. 
\begin{lemma}
\label{lem:D_f-LLR-bi-affine}
    Let $\mu_1, \mu_2$, and $\mu_3$ be LLR distributions and let $p \in [0,1]$. Then
    \begin{equation}
    \label{eq:d_affinity}
        D_f\bigl((p\,\mu_1 + (1-p)\,\mu_2)*\mu_3\bigr) = p\, D_f(\mu_1 *\mu_3) + (1-p)\,D_f(\mu_2 * \mu_3).
    \end{equation}
    \begin{IEEEproof}
        From Definition \ref{def:convolution}, it is straightforward to check that
        \begin{equation}
            \label{eq:distributivity-mixture-convolution}
            (p\,\mu_1 + (1-p)\,\mu_2)*\mu_3 = p\,(\mu_1 *\mu_3) + (1-p)\, (\mu_2 * \mu_3).
        \end{equation}
    Moreover, the maps $\mu \mapsto \int \exp^- \,d\mu$ and $\mu \mapsto \int f \circ \exp^- \,d\mu$ are affine with respect to mixtures of LLR distributions. The former is a standard fact since $\exp^-$ is a non-negative function. Regarding the latter, because $D_f(\mu)$ is well-defined and non-negative for any LLR distribution $\mu$, then the integration of the negative part of $f \circ \exp^-$ is finite for any such $\mu$. Hence, $\mu \mapsto \int f \circ \exp^- \,d\mu$ is affine with respect of mixtures of LLR distributions as well. The proof concludes by invoking the preceding after substituting \eqref{eq:distributivity-mixture-convolution} into \eqref{eq:def_f-divergence_LLR}.
    \end{IEEEproof}
\end{lemma}

We are ready to show that every section of a tensorization base is an affine map. The key idea behind this proof is that tensorization upgrades the fact that $D_f(\mu_1 * \mu_2)$ is affine as a function of $\mu_1$ to being affine as a function of $D_f(\mu_1)$ because $D_f(\mu_1 * \mu_2)$ is completely determined by the two scalars $D_f(\mu_1)$ and $D_f(\mu_2)$. 
\begin{lemma}
\label{lem:base-is-affine-fn}
    Let $\tau_f$ be a tensorization base. For any $p \in [0,1]$ and $d_1, d_2, d_3 \in \mathcal{R}(D_f)$
    \begin{equation}
        \tau_f(p\, d_{1} + (1-p) d_{2}, d_3) = p\, \tau_f(d_{1}, d_3) + (1-p)\,\tau_f(d_{2}, d_3) = \tau_f(d_3, p\, d_{1} + (1-p) d_{2}).
    \end{equation}
    \begin{IEEEproof}
        Let $d_1, d_2, d_3 \in \mathcal{R}(D_f)$. By definition of $\mathcal{R}(D_f)$, there exist LLR distributions $\mu_1, \mu_2, \mu_3$ such that $d_i = D_f(\mu_i)$, for $i \in \{1,2,3\}$. Then, 
        \begin{align*}
            \tau_f(p\, d_{1} + (1-p) d_{2}, d_3) &= \tau_f(D_f(p\, \mu_1 + (1-p) \mu_2), D_f(\mu_3)) \\
            &= D_f((p\, \mu_1 + (1-p) \mu_2)* \mu_3) \quad \text{(by Lemma~\ref{lemma:Df-df})} \\
            &= p\, D_f(\mu_1 * \mu_3) + (1-p)D_f(\mu_2*\mu_3) \quad \text{(by Lemma~\ref{lem:D_f-LLR-bi-affine})}\\
            &= p\, \tau_f(d_1, d_3) + (1-p)\tau_f(d_2, d_3).
        \end{align*}
        Finally, the fact that $d_2 \mapsto \tau_f(d_1, d_2)$ is affine follows from symmetry (Lemma~\ref{lem:tensorization-base-properties}). 
        \end{IEEEproof}
\end{lemma}

We conclude this section with a classification of admissible tensorization bases. Once Lemma~\ref{lem:base-is-affine-fn} shows that $\tau_f$ is affine in each argument, the marginalization property from Lemma~\ref{lem:tensorization-base-properties} forces $\tau_f$ to take a ``multi-affine'' form.

\begin{theorem}
\label{thm:functional_form_base}
Let $D_f$ admit $\tau_f$ as its tensorization base. There exists $\gamma$ in $\mathbb{R}$ such that
\begin{equation} \label{eq:tensorization_base_functional_form}
    \tau_f(d_1, d_2) = d_1 + d_2 + \gamma\,d_1 d_2, \quad \forall d_1, d_2 \in \mathcal{R}(D_f).
\end{equation}
We call $\gamma$ the \emph{tensorization parameter} of $D_f$.
\begin{IEEEproof}
    First, we consider the degenerate case where $\mathcal{R}(D_f) = \{0\}$, which arises for any $f$ of the form $x\mapsto B(x-1)$ for some $B \in \mathbb{R}$. In this setting, \eqref{eq:tensorization_base_functional_form} trivially holds for any $\gamma \in \mathbb{R}$. Now, assume $\mathcal{R}(D_f) \neq \{0\}$. Invoking Lemma \ref{lem:base-is-affine-fn}, for a fixed $d_2 \in \mathcal{R}(D_f)$, $d_1\mapsto \tau_f(d_1, d_2)$ is an affine map. Thus, there exist $\alpha(d_2)$ and $\beta(d_2)$ independent of $d_1$ such that 
    \begin{equation}
    \label{eq:tau-affine-1}
        \tau_f(d_1, d_2) = \alpha(d_2)\, d_1 + \beta(d_2), \quad \forall d_1 \in \mathcal{R}(D_f).
    \end{equation}
    Marginalization (Lemma \ref{lem:tensorization-base-properties}) implies $\beta(d_2) = d_2$. Now, fix $d_1 \neq 0$, then, in general, for any $d_2 \in \mathcal{R}(D_f)$:
    \begin{equation*}
        \label{eq:affine-base-alpha}
        \alpha(d_2) = \frac{\tau_f(d_1, d_2) - d_2}{d_1}.
    \end{equation*}
    It is immediate to check that $d_2 \mapsto \frac{\tau_f(d_1, d_2)-d_2}{d_1}$ is an affine map and that $\alpha(0) = 1$ due to the marginalization property of $\tau_f$. Hence, there exists $\gamma \in \mathbb{R}$ for which 
    $\alpha(d_2) = \gamma d_2 + 1$. Substituting $\alpha(d_2) = \gamma d_2 + 1$ and $\beta(d_2) = d_2$ in \eqref{eq:tau-affine-1} concludes the proof.  
\end{IEEEproof}
\end{theorem}

\section{Symmetrized Tensorizable $f$-divergences}
\label{sec:local-characteristic}
In this section, we characterize tensorization on the subclass of symmetrized $f$-divergences. As we will show, the symmetrization $D_f^{\textnormal{sym}}$ of a tensorizable $f$-divergence is uniquely determined by two scalar parameters: the parameter $\gamma$ in its tensorization base and a quantity that we call its \emph{local characteristic}. Intuitively, the local characteristic captures the quadratic rate at which $D_f^{\textnormal{sym}}$ vanishes as two probability measures become indistinguishable.

Our proof technique involves the restriction of $D_f^{\textnormal{sym}}$ to the family of two-point symmetric LLR distributions. On this family the $f$-divergence simplifies to a real-valued function $g$, and tensorization enforces a ``doubling recursion'' for $g$, which makes the map $h:r \mapsto g(r)/r^2$ approximately invariant across a halving scale; specifically, it implies that $h(r) \approx h(r/2)$ for small enough $r$, and that the limit of $h$ as $r \to 0^+$ always exists. This limit is the local characteristic. 
We show that if two symmetrized $f$-divergences $D_{f_1}^{\textnormal{sym}}$ and $D_{f_2}^{\textnormal{sym}}$ share the same local characteristic and the same tensorization base, then their respective restrictions to two-point symmetric LLR distributions $g_1$ and $g_2$ must be identical on a neighborhood around zero; if $g_1$ and $g_2$ disagreed, the approximate scale-invariance mentioned above would propagate their discrepancy as $r \to 0$, implying convergence to different local characteristics, a contradiction. This local equivalence propagates to the full domain of $g_1$ and $g_2$ through the ``doubling recursion'' induced by tensorization. Given that the restriction $g$ uniquely identifies $D_f^{\textnormal{sym}}$, the previous argument yields a characterization of symmetrized tensorizable $f$-divergences in terms of their local characteristic and their tensorization parameter $\gamma$.

We begin our formal argument with the definition of a symmetric LLR distribution, which we adopt from \cite{yu2023ising}.
\begin{definition}[Symmetric LLR distribution]
We say that an LLR distribution $\mu$ is symmetric if, for every $A \in \mathcal{B}(\mathbb{R}_{+\infty})$, 
    \begin{equation}        \label{eq:def_symmetric_LLR}
        \int_A d\mu(r) = \int \mathbf{1}_A(-r) e^{-r} d\mu(r). 
    \end{equation}   
\end{definition}
As we will show, tensorization on symmetrized $f$-divergences can be characterized using two-point symmetric LLR distributions. We introduce these two notions next.

\begin{definition}
For a convex function $f: (0, \infty) \to \mathbb{R}$, let $\Tilde{f}: x \mapsto x f(1/x)$. We define the symmetrized $f$-divergence $D^{\textnormal{sym}}_f$ as 
\begin{equation}
    \label{eq:symmetrized-D_f}
    D^{\textnormal{sym}}_f(P, Q) = \frac{1}{2}D_f(P \Vert Q) + \frac{1}{2}D_f(Q \Vert P) = \frac{1}{2}D_{f+\tilde{f}}(P \Vert Q)
\end{equation}
Following our previous convention for $D_f$, we may write 
\[
D^{\textnormal{sym}}_f(\mu) = \frac{1}{2}D_{f+\tilde{f}}(\mu),
\]
where $\mu = \hat{\mu}(P \Vert Q)$ for some probability measures $P$ and $Q$.
\end{definition}

\begin{proposition}
    \label{ex:symmetric-2-LLR}
    Let $r >0$. The unique symmetric LLR distribution $\mu_r$ with support $\{-r, r\}$ is given by 
    \begin{equation}
        \mu_r(\{r\}) = \frac{1}{1+e^{-r}},\quad \mu_r(\{-r\}) = \frac{e^{-r}}{1 + e^{-r}}.
    \end{equation}
    Additionally, we define $\mu_0$ as the probability measure with support $\{0\}$.
\end{proposition}

\begin{remark}
If $\mu$ is a symmetric LLR distribution with $\int \exp^{-} d\mu = 1$, then $D^{\textnormal{sym}}_f(\mu) = D_f(\mu).$ This property will later serve as a bridge to obtain a complete characterization of $D_f$ from this section's results.
\end{remark}

After restricting $D_f^\text{sym}$ to the family of two-point symmetric LLR distributions, $D_f^\text{sym}$ is identical to $D_f$ and tensorization induces a classical functional equation involving a real function. In particular, if $D_f$ tensorizes, there exists $\gamma \in \mathbb{R}$ such that
\begin{equation}
\label{eq:Df-r1-r2}
    D_f^\text{sym}(\mu_{r_1}*\mu_{r_2}) = D_f^\text{sym}(\mu_{r_1}) + D_f^\text{sym}(\mu_{r_2}) + \gamma D_f^\text{sym}(\mu_{r_1})D_f^\text{sym}(\mu_{r_2}).
\end{equation}
An explicit computation of this expression yields the following quadratic functional equation.
\begin{lemma}
    \label{lem:g-r1+r2}
     Let $D^{\textnormal{sym}}_f$ be a symmetrized tensorizable $f$-divergence, and let $g(r)$ denote its evaluation at the symmetric LLR distribution $\mu_r$. There exists $\gamma \in \mathbb{R}$ such that for any $0 \leq r_2 \leq r_1$
    \begin{equation}
    \label{eq:key-equation}
        \frac{1+e^{-r_1-r_2}}{(1+e^{-r_1})(1+e^{-r_2})}g(r_1+r_2) + \frac{1+e^{r_2-r_1}}{e^{r_2}(1+e^{-r_1})(1+e^{-r_2})}g(r_1-r_2) = g(r_1) + g(r_2) + \gamma g(r_1)g(r_2).  
    \end{equation}
\end{lemma}

We characterize $D_f^{\textnormal{sym}}$ by solving Equation~\eqref{eq:key-equation}. This equation is sufficiently constraining that a single local parameter fully describes its solution space: the quadratic rate at which $g(r)$ vanishes as $r \to 0$. We first ensure that this rate is always well-defined and finite for any $D_f^{\textnormal{sym}}$. 
\begin{theorem}
\label{thm:sym_c}
Let $D^{\textnormal{sym}}_f$ be any symmetrized tensorizable $f$-divergence, and let $g(r)$ denote its evaluation at the symmetric LLR distribution $\mu_r$. The limit 
\begin{equation}
    \label{eq:local-characteristic}
    \lim_{r\to 0^+} \frac{g(r)}{r^2}
\end{equation}
always exists, is finite and nonnegative. We refer to this value as the \emph{local characteristic} of $D^{\textnormal{sym}}_f$ and of $D_f$. More generally, for any function $g : [0,\infty) \to [0,\infty)$, whenever the limit in \eqref{eq:local-characteristic} exists and is finite, we call it the local characteristic of $g$.
\end{theorem}
We provide the proof of Theorem~\ref{thm:sym_c} in Appendix~\ref{appendix:proof-local-characteristic}. 

    Theorem \ref{thm:sym_c} immediately rules out the Total Variation (TV) distance as a tensorizable divergence. Recall that TV is the $f$-divergence defined by $f: t \mapsto \frac{1}{2}\lvert t - 1 \rvert$. In this case, $g(r) = \frac{1-e^{-r}}{1+e^{-r}} = \tanh(r/2)$. Invoking L'Hôpital's rule, we can see its local characteristic diverges: 
    \[
        \lim_{r \to 0^+} \frac{1-e^{-r}}{(1+e^{-r})r^2} = \lim_{r \to 0^+} \frac{e^{-r}}{2r(1+e^{-r}) - r^2e^{-r}} = +\infty. 
    \]

The next result shows that the local characteristic is not merely a descriptive quality of $g$. Together with the tensorization parameter $\gamma$, it determines all solutions to the functional equation in \eqref{eq:key-equation}. A proof of this result can be found in Appendix~\ref{appendix:proof-of-unique-g}.

\begin{theorem}
\label{thm:unique-g}
    Fix a tensorization parameter $\gamma \in \mathbb{R}$ (as defined in Theorem~\ref{thm:functional_form_base}) and a local characteristic $\ell \in [0,\infty)$. There is at most one continuous function $g: [0,\infty) \to [0,\infty)$ with this local characteristic that also solves \eqref{eq:key-equation} for this particular $\gamma$.
\end{theorem}

We now translate uniqueness of the restriction $g$ into a uniqueness of the symmetrized $f$-divergence. Specifically, a change of variables reveals that $g$ characterizes the map $x \mapsto f(x) + x f(1/x)$, which uniquely identifies $D_f^{\textnormal{sym}}$.
\begin{theorem}[Uniqueness of Symmetrized $f$-divergence]
\label{def:unique_sym}
For any tensorization base $\tau$ and any local characteristic $\ell$, there exists at most one symmetrized $f$-divergence $D^{\textnormal{sym}}_f$ that has local characteristic $\ell$, whose associated $f$-divergences $D_f$ and $D_{\tilde{f}}$ admit $\tau$ as their tensorization bases.
\begin{IEEEproof}
    Let $\tau$ be a tensorization base and $\ell \in [0, \infty)$. The statement trivially holds if there exists no tensorizable $D^{\textnormal{sym}}_f$ that has local characteristic $\ell$ and tensorization base $\tau$ associated with $D_f$ and $D_{\tilde{f}}$. Now, assume there exists at least one $D^{\textnormal{sym}}_f$ with these two features. Moreover, suppose $D_{f_1}$ and $D_{f_2}$ have tensorization base $\tau$ and local characteristic $\ell$. Let $g_1(r)$ and $g_2(r)$ be the evaluations of $D^{\text{sym}}_{f_1}$ and $D^{\text{sym}}_{f_2}$ at the symmetric LLR distribution on $\{-r,r\}$, with $r \geq 0$. A simple substitution yields 
    \begin{equation}
        g_i(r) = f_i(e^{-r}) \cdot \frac{1}{1 + e^{-r}} + f_i(e^{r}) \cdot \frac{e^{-r}}{1 + e^{-r}}, \quad i \in \{1,2\}. 
    \end{equation}
    By Theorem~\ref{thm:unique-g}, $g_1(r) = g_2(r)$ for all $r \in [0,\infty)$. Thus, by considering the substitutions $x = e^{-r}$ and $x = e^r$, we get
    \begin{equation}
    \label{eq:symmetrized-f-equality}
        (f_1+\tilde{f}_1)(x) = (f_2+\tilde{f}_2)(x), \quad \forall x \in (0,\infty).
    \end{equation}
    Invoking \eqref{eq:symmetrized-D_f}, this implies $D_{f_1}^{\textnormal{sym}}=D_{f_2}^{\textnormal{sym}}$.
\end{IEEEproof}
\end{theorem}

By leveraging the uniqueness ensured by the previous theorem, we next list all solutions to \eqref{eq:key-equation}.   
\begin{corollary}
\label{cor:explicit-form-g}
    Suppose $g: [0, \infty) \to [0, \infty)$ solves \eqref{eq:key-equation} for some $\gamma \in \mathbb{R}$ and $g$ has a local characteristic $\ell \in [0, \infty)$. Then, for all $r \in [0,\infty)$, 
    \begin{equation}
        \label{eq:explicit-form-g}
        g(r) = \begin{cases}
            \displaystyle \frac{1}{\gamma}\left(\frac{e^{-\alpha r} + e^{-(1-\alpha)r}}{1+e^{-r}}-1\right), & \gamma \neq 0,\, \gamma\ell \geq -\frac{1}{8}, \\
            \displaystyle \frac{1}{\gamma}\left(\frac{\cos(\beta r)}{\cosh(r/2)}-1\right), &  \gamma\ell < -\frac{1}{8},\\
            \displaystyle \frac{2\ell r(1-e^{-r})}{1+e^{-r}}, & \gamma = 0,
        \end{cases}
    \end{equation}
    where $\alpha$ is any root of the polynomial $x(x - 1) - 2\gamma \ell$ and $\beta = \frac{1}{2}\sqrt{-1-8\gamma \ell}$.
    \begin{IEEEproof}
        Let $\gamma \in \mathbb{R}$ and $\ell \in [0, \infty)$ and take $g$ according to \eqref{eq:explicit-form-g}. It is straightforward to check that $g$ fulfills equation \eqref{eq:key-equation} for this $\gamma$ and that \[\lim_{r \to 0^+}\frac{g(r)}{r^2} = \ell.\]
        By Theorem~\ref{thm:unique-g}, any other function from $[0, \infty)$ to $[0,\infty)$ that also solves \eqref{eq:key-equation} with this parameter $\gamma$ and that has local characteristic $\ell$ needs to be identical to $g$.
    \end{IEEEproof}
\end{corollary} 

Theorem~\ref{thm:unique-g} and Corollary~\ref{cor:explicit-form-g} apply to any solution of Equation~\eqref{eq:key-equation} with a local characteristic. However, our next lemma states that only two of the three regimes in Corollary~\ref{cor:explicit-form-g} are admissible when $g$ comes from a symmetrized tensorizable $f$-divergence. In this case, because $g$ is induced by a convex function, it is straightforward to prove that it is a nondecreasing function. However, when $\gamma \ell < -\frac{1}{8}$, the corresponding expression in \eqref{eq:explicit-form-g} is oscillatory, so it cannot be nondecreasing. 

\begin{theorem}
\label{thm:relevant-ell-gamma-case}
    Let $D_f^{\text{sym}}$ be a symmetrized tensorizable $f$-divergence, $\gamma$ its tensorization parameter and $\ell$ its local characteristic. We have $\gamma \ell \geq -\frac{1}{8}$.
\end{theorem}

\section{Characterization of Tensorizable $f$-divergences}
\label{sec:main-characterization}
In this section we present our main characterization of tensorizable $f$-divergences. Our method consists of evaluating a tensorizable $f$-divergence on the convolution of an arbitrary two-point LLR distribution and a symmetric one to obtain a nonhomogeneous linear functional equation satisfied by any function that induces the $f$-divergence. Our main characterization then follows by identifying all convex solutions of this equation. The next lemma provides the underlying identity in this analysis. The result follows directly from Theorem~\ref{thm:functional_form_base} and the fact that $D_f(\mu) = D_f^{\textnormal{sym}}(\mu)$ for any symmetric LLR distribution $\mu$. 
\begin{lemma}
Let $D_f$ be an $f$-divergence with tensorization parameter $\gamma$. If $\nu$ and $\mu$ are, respectively, an arbitrary and a symmetric LLR distribution, then
 \begin{equation}
        \label{eq:semi-symmetrized}
        D_f(\nu * \mu) = D_f(\nu) + D^{\textnormal{sym}}_f(\mu) +\gamma D_f(\nu) D^{\textnormal{sym}}_f(\mu). 
\end{equation}
\end{lemma}
\begin{definition}
    We say that a probability measure $\mu$ is a two-point LLR distribution with support $\{s_1, s_2\} \subset \mathbb{R}$, $s_1 < 0 < s_2$, if 
    \[
    \mu(\{s_1\}) = \frac{1-e^{-s_2}}{e^{-s_1}-e^{-s_2}}, \quad \mu(\{s_2\}) = \frac{e^{-s_1}-1}{e^{-s_1}-e^{-s_2}}.
    \]
    Note that taking $s_1 = -r$, $s_2 = r$ and $r >0$ yields a symmetric LLR distribution.
\end{definition}
Setting $\nu$ and $\mu$ to be the two-point LLR distributions with support $\{s_1,s_2\}$ and $\{-r,r\}$, respectively, in \eqref{eq:semi-symmetrized} yields an analogue of Equation~\eqref{eq:key-equation}. To express this equation succinctly we introduce the following notation. Recall that $f$ denotes a convex function, $D_f^{\text{sym}}$ its induced symmetrized $f$-divergence, and $g$ the restriction of $D_f^{\text{sym}}$ to symmetric LLR distributions. We define the transformation $\phi$ as
\begin{equation}
\label{eq:linear_form}
    \phi^\gamma_g[f](s_1, s_2, r) \coloneqq  \frac{\rho(s_1)T_{r}[f](e^{-s_2})-\rho(s_2)T_{r}[f](e^{-s_1})}{\rho(s_1)- \rho(s_2)} - [1+\gamma g(r)] \frac{\rho(s_1)f(e^{-s_2}) - \rho(s_2)f(e^{-s_1})}{\rho(s_1)- \rho(s_2)},
\end{equation}
where 
\[
    T_r[f](x) \coloneqq \frac{f(xe^{-r}) + e^{-r}f(xe^r)}{1+e^{-r}}, \quad x \in (0,\infty), r \in [0, \infty),
\] 
and
\[
    \rho(s) \coloneqq e^{-s}-1, \quad s \in \mathbb{R}.
\] 
For a fixed $g$, $\phi_g^\gamma$ is a linear map from the space of continuous functions to the space of real-valued functions of $(s_1, s_2, r)$. This motivates us to consider \eqref{eq:semi-symmetrized} as a nonhomogeneous linear equation in $f$ for a fixed $g$. 
\begin{lemma}
Let $f: (0,\infty) \to \mathbb{R}$ define an $f$-divergence with tensorization parameter $\gamma$ and let $g: [0,\infty) \to [0,\infty)$ be its restriction to two-point symmetric LLR distributions. For any $s_1 < 0 < s_2$ and $r \geq 0$, we have
    \begin{equation}
    \label{eq:another-key-equation}
    \phi^\gamma_g[f](s_1, s_2, r)
    = g(r).
\end{equation}
\end{lemma}

Remarkably, by the definition of symmetrized $f$-divergence, we already have
    \begin{align}
      T_r[f](1)=g(r).
\end{align}
Although this identity will be a convenient intermediate step in the proof, it is already implied by \eqref{eq:another-key-equation} when $f(1)=0$ and therefore need not be imposed as a separate condition.
 Accordingly, the rest of this section is concerned with identifying the solutions of \eqref{eq:another-key-equation} for fixed $g$ and $\gamma$. 
 
 Our next result presents a basis for $\ker \phi_g^\gamma$. The intuition behind our proof is that, for a fixed $g$, $\ker \phi_g^\gamma$ has dimension at most two. Once this upper bound is established, it suffices to construct two linearly independent elements of the kernel. Such elements can be obtained from functions already known to induce tensorizable $f$-divergences. We focus on the regime where the local characteristic $\ell$ and the tensorization parameter $\gamma$ satisfy $\ell \gamma \geq -\frac{1}{8}$ since this is the only regime admissible when $g$ is induced by a tensorizable $f$-divergence (see Theorem~\ref{thm:relevant-ell-gamma-case}).

\begin{theorem}
\label{thm:basis-kernel}
    Let $\gamma \in \mathbb{R}$, and take $g: [0,\infty) \to [0, \infty)$ as in Corollary~\ref{cor:explicit-form-g} with local characteristic $\ell \in [0,\infty)$ such that 
    \[
    \ell \gamma \geq -\frac{1}{8}.
    \]
    Consider $\phi_g^\gamma$ as a linear transformation on the vector space of continuous functions $h:(0, \infty) \to \mathbb{R}$ with $h(1) = 0$. Then, $\ker \phi_g^\gamma$ has dimension two. Furthermore, 
    \begin{equation}
    \label{eq:bases}
        \ker \phi_g^\gamma = \begin{cases}
            \operatorname{span}\left(\{x-1, (x+1)\log x\}\right), & \gamma \ell = 0, \\
            \operatorname{span}\left(\{x-1, \sqrt{x}\log x\}\right), & \gamma \ell = -\frac{1}{8}, \\
            \operatorname{span}(\{x-1, x^{\alpha} - x^{1-\alpha}\}), &  \gamma \ell > -\frac{1}{8}, \gamma\ell \neq 0, 
        \end{cases}
    \end{equation}
    where, in the last case, $\alpha$ is either root of the polynomial $x(x - 1) - 2\gamma \ell$.
\end{theorem}

Since $\phi_g^\gamma$ is linear in $f$, the difference between any two solutions of \eqref{eq:another-key-equation} lies in $\ker \phi_g^\gamma$. Hence, after fixing a particular solution $f_0$ to \eqref{eq:another-key-equation}, we can decompose any other solution $f$ as 
\[
    f = f_0 + h, \quad h \in \ker \phi_g^\gamma. 
\]
Combining this observation with Theorem~\ref{thm:basis-kernel}, we can identify all continuous solutions $f$ to \eqref{eq:another-key-equation} satisfying $f(1) = 0$, for any fixed local characteristic $\ell$ and tensorization parameter $\gamma$ such that $\ell \gamma \geq - \frac{1}{8}$. Note that, in light of Lemma~\ref{thm:relevant-ell-gamma-case}, this is the only admissible regime that can arise when $\ell$ and $\gamma$ come from a tensorizable $f$-divergence. Because every convex function $f$ that induces a tensorizable $D_f$ satisfies \eqref{eq:another-key-equation}, it only remains to impose convexity on these solutions to obtain the complete characterization below. 
\begin{theorem}
    \label{thm:characterization-tensorization}
    Let $D_f$ be an $f$-divergence with tensorization parameter $\gamma \in \mathbb{R}$ and local characteristic $\ell \in [0 ,\infty)$. If $\ell = 0$, there exists a real number $E$ such that
    \begin{equation}
        \label{eq:affine-f-solution}
        f(x) = E(x-1),\qquad x\in(0,\infty).
    \end{equation}
    Suppose now that $\ell >0$. If $\gamma = 0$, then there exist $A, B \in [0,\infty)$ and $E \in \mathbb{R}$ such that
    \begin{equation}
        \label{eq:additive-f-divergence}
        f(x)=A x\log x-B\log x + E(x-1),
        \quad x\in(0,\infty),
    \end{equation}
    where $A + B = 2\ell$. Otherwise, there exist $\alpha, E\in\mathbb{R}$ such that
    $\alpha(\alpha-1) = 2\gamma\ell$ and
    \begin{equation}
    \label{eq:non-additive-f-divergence}
        f(x)=\frac{x^\alpha-1}{\gamma} + E(x-1), \quad x \in (0,\infty).
    \end{equation}
\end{theorem}

The expressions for $f$ in Theorem~\ref{thm:characterization-tensorization} show every tensorizable $f$-divergence is either a linear combination, with nonnegative coefficients, of the KL and reverse KL divergences or a power divergence. We refer the reader to Appendix~\ref{appendix:sols-characterization-semi-sym} for the proofs of Theorems~\ref{thm:basis-kernel} and~\ref{thm:characterization-tensorization}. 

\section{Conclusion}
In this work, we characterized every tensorizable $f$-divergence in the sense of Definition~\ref{def:tensorization_base}. In particular, we showed that assuming the total discrepancy 
\[
D_f(P_1 \otimes \cdots \otimes P_n \Vert Q_1 \otimes \cdots \otimes Q_n)
\] 
decomposes solely as a function of $\{D_f(P_i \Vert Q_i)\}_{i = 1}^n$ forces $D_f$ to be either a linear combination, with nonnegative coefficients, of the KL and reverse KL divergences or a power divergence.

There have been several successful and general characterizations of additive functionals on spaces of probability measures \cite{mattner1999cumulants, mu2024monotone, MuXiaosheng2021FBDi}. For instance, Theorem 2 in \cite{MuXiaosheng2021FBDi}, together with Theorem~\ref{thm:functional_form_base}, yields an integral representation for any tensorizable $f$-divergence. However, our characterization does not follow immediately from the integral representation in \cite[Theorem 2]{MuXiaosheng2021FBDi}. We also emphasize that the definition of tensorization that we adopt is not framed in terms of additivity. Rather, one of our main contributions is to connect the formalism in \cite{CruzRodrigo2025Tof} to the more classical analysis of pseudo-additive functionals.

Additionally, the axiomatic characterization of information measures has a long history of identifying functionals of probability measures that satisfy compositional axioms through the study of functional equations \cite{HobsonArthur1969Anto, KannappanPl1973Msof, kannappan1974functional, csiszar2008axiomatic, ebanks1998characterization}. Hence, our analysis naturally overlaps with this literature once tensorization is distilled into a functional equation. However, the focus of this work is the tensorization property itself, the compositional rules it induces, and which \(f\)-divergences admit this property. To our knowledge, a self-contained and direct characterization of tensorization within the class of \(f\)-divergences is missing from the current literature, and providing such a characterization is the main contribution of this work.

Finally, we note that the class of tensorizable $f$-divergences is not closed under addition, even though $D_{f+g} = D_f + D_g$ defines a valid $f$-divergence and simplifies under product measures if $D_f$ and $D_g$ tensorize. However, Definition~\ref{def:tensorization_base} does not, in general, regard $D_{f+g}$ as tensorizable because its simplification rule depends on all marginal discrepancies $\{D_f(P_i \Vert Q_i)\}_{i = 1}^n$ and $\{D_g(P_i\Vert Q_i)\}_{i = 1}^n$ rather than only depending on $\{D_{f+g}(P_i\Vert Q_i)\}_{i = 1}^n$. Nevertheless, our notion of tensorization can serve as a building block for studying broader classes of $f$-divergences with compositional properties, such as linear combinations of tensorizable $f$-divergences. Characterizing such classes is a natural direction for future work.


\appendices
\section{The $f$-divergence and the Log-likelihood Ratio Distribution}
\label{appendix:measure-theoretic-details}
In this appendix, we present the measure-theoretic details that arise in our study of $f$-divergences using the LLR distribution. To begin with, we present a chain rule for Radon-Nikodym derivatives. This result ensures that the $f$-divergence and the log-likelihood ratio distribution are invariant to the choice of dominating measure for a pair $(P,Q)$. Note that we write $P \ll Q$ to denote that $P$ is absolutely continuous with respect to $Q$.
\begin{lemma}[\cite{FollandG.B1984Ra}]
\label{lem:Folland-chain-rule}
    Suppose that $\nu$,
    $\mu$, and $\lambda$ are $\sigma$-finite measures on the same measurable space with $\nu \ll \mu$ and $\mu \ll \lambda$. We have $\nu \ll \lambda$, and 
    \begin{equation}
        \label{eq:Folland}
        \frac{d\nu}{d\lambda} = \frac{d\nu}{d\mu} \frac{d\mu}{d\lambda} \quad \lambda - \text{a.e.} 
    \end{equation}
\end{lemma}

In this subsection, we provide the details needed to migrate from the $f$-divergence framework in Definition~\ref{def:f-divergence} to the log-likelihood ratio regime. In the following, we show that Definition \ref{def:LLR-distribution-RN} is sound; i.e., that $\hat{\mu}(P \Vert Q)$ is independent of the choice of dominating measure $\lambda$ using the chain rule for Radon-Nykodim derivatives (Lemma~\ref{lem:Folland-chain-rule}).

\begin{lemma}[Invariance of the log-likelihood ratio]
\label{lem:invariance_LLR_as}
Let $P$ and $Q$ be probability measures on the same measurable space $(E, \mathcal{E})$. Assume there exist $\sigma$-finite measures $\lambda_1, \lambda_2$ on $(E, \mathcal{E})$ such that $P,Q \ll \lambda_1$ and $P,Q \ll \lambda_2$, then, under the convention that $a/0 = +\infty$ for any $a \in (0, \infty)$, 
\begin{equation}
    \label{eq:equal_LLR}
    Q\left(\frac{q_1}{p_1} = \frac{q_2}{p_2}\right) = 1,
\end{equation}
where 
\[p_i \coloneqq \frac{dP}{d\lambda_i}, \quad q_i \coloneqq \frac{dQ}{d\lambda_i}, \;\, i \in \{1,2\}.\]
\begin{IEEEproof}
    Initially, assume that $\lambda_1 \ll \lambda_2$. Invoking Lemma \ref{lem:Folland-chain-rule}, 
    \[
    \frac{dP}{d\lambda_2} = \frac{dP}{d\lambda_1} \frac{d\lambda_1}{d\lambda_2}, \quad \frac{dQ}{d\lambda_2} = \frac{dQ}{d\lambda_1} \frac{d\lambda_1}{d\lambda_2} \quad \lambda_2 - \text{a.e.}
    \]
    The relation $Q \ll \lambda_2$ implies that, for any sets $E_1, E_2 \in \mathcal{E}$, if $\lambda_2(E_1) = 0$ and $\lambda_2(E_2) = 0$, then $Q(E_1) = 0= Q(E_2)$. Hence, the previous two equations hold $Q \;-$ a.s. Moreover, $\frac{d\lambda_1}{d\lambda_2} >0$, $Q\;-\text{a.s.}$, which follows directly from the fact that $\{\frac{d\lambda_1}{d\lambda_2} = 0 \}$ is a $\lambda_1$-null set and that $Q \ll \lambda_1$.
    Thus, $Q\; -$ a.s.,
    \begin{align*}
        \frac{q_1}{p_1} &= \frac{q_1 \, \frac{d\lambda_1}{d\lambda_2}}{p_1 \, \frac{d\lambda_1}{d\lambda_2}}\\
        &= \frac{q_2}{p_2},
    \end{align*}
    where the division is well-defined  since $Q(q_1=0) =0$. 

    If $\lambda_1 \not\ll \lambda_2$ and $\lambda_2 \not\ll \lambda_1$, we define an auxiliary measure $\lambda^* = \frac{1}{2}(\lambda_1 + \lambda_2)$ and follow the same approach as above to obtain the same conclusion.
\end{IEEEproof}
\end{lemma}

From Lemma~\ref{lem:invariance_LLR_as} follows that the LLR distribution $\hat{\mu}(P \Vert Q)$ is invariant to the chosen dominating measure $\lambda$.
\begin{corollary}[Invariance of the LLR distribution]\label{corollary:invariance_LLR_distribution}
    Let $P$ and $Q$ be two probability measures on the same measurable space $(E, \mathcal{E})$. 
    For any $\sigma$-finite measures $\lambda_1, \lambda_2$ on $(E, \mathcal{E})$ such that $P,Q \ll \lambda_1$ and $P,Q \ll \lambda_2$, we have 
    \begin{equation}
        \label{eq:equal_LLR_distr}
        Q \circ \left(\log \frac{q_1}{p_1}\right)^{-1} = Q\circ \left(\log \frac{q_2}{p_2}\right)^{-1},
    \end{equation}
    where $p_1, q_1, p_2, q_2$ are defined as in Lemma \ref{lem:invariance_LLR_as}.
    \begin{IEEEproof}
        It is widely known that, if two random variables are almost surely equal, then they share the same distribution. In our setting this fact translates to: if $\log \frac{q_1}{p_1} = \log \frac{q_2}{p_2}$, $Q\;- $ a.s., then 
        \[
        Q\left(\log \frac{q_1}{p_1} \in A\right) =  Q\left(\log \frac{q_2}{p_2} \in A\right), \quad \forall A \in \mathcal{B}(\mathbb{R}_{+\infty}).
        \]
        We finish the proof by noting this last statement is just another way of writing Equation \eqref{eq:equal_LLR_distr}.
    \end{IEEEproof}
\end{corollary}

We refer to all measures of the form $\hat{\mu}(P \Vert Q)$, for some $P$ and $Q$, as the class of all \emph{realizable} LLR distributions and denote it by $\mathscr{L}$. As reflected in previous literature (see, for example, \cite{yu2023ising}), $\mathscr{L}$ has a simpler characterization, namely: 
\begin{equation}
    \label{eq:set_of_LLR_dists}
    \mathscr{L} = \left\{\text{$\nu$ on $(\mathbb{R}_{+\infty}, \mathcal{B}(\mathbb{R}_{+\infty}))$}: \int \exp^- \, d\nu \leq 1\right\}.
\end{equation}
We provide a proof of this fact in our next lemma, and proceed to use this representation in the proof of Lemma~\ref{lem:L-closed-under-operations}.
\begin{lemma}
\label{lem:LLR-equivalent-defs}
    Given a probability measure $\nu$ on $(\mathbb{R}_{+\infty}, \mathcal{B}(\mathbb{R}_{+\infty}))$, there exist a measurable space $(E, \mathcal{E})$, probability measures $P$ and $Q$ on $(E, \mathcal{E})$ and a dominating measure $\lambda$ such that $P,Q \ll \lambda$ and 
    \[
    \nu = \left.\hat{\mu}(P \Vert Q)\right|_{\mathcal{B}(\mathbb{R}_{+\infty})} = \left.Q\circ \left(\log \frac{q}{p} \right)^{-1}\right|_{\mathcal{B}(\mathbb{R}_{+\infty})}.
    \]
    if and only if 
    \begin{equation}
        \label{eq:ineq-LLR-def}
        \int \exp^{-} \, d\nu \leq 1. 
    \end{equation}
    \begin{IEEEproof}
        Let $\nu$ be a probability measure on $(\mathbb{R}_{+\infty}, \mathcal{B}(\mathbb{R}_{+\infty}))$. On the one hand, assume that there exists a measurable space $(E, \mathcal{E})$, probability measures $P$ and $Q$ on $(E, \mathcal{E})$ and a dominating measure $\lambda$ such that $P,Q \ll \lambda$ and  $\nu = Q\circ \left(\log \frac{q}{p} \right)^{-1}$. Then, by definition of the push-forward measure $\nu$,
        \begin{align*}
            \int \exp^- \, d\nu &= \int_{\{q > 0\}} p \,d\lambda \\
            &= P(\{q > 0\}).
        \end{align*}
        Thus, $\int \exp^- \, d\nu \leq 1$. On the other hand, assume that $\int \exp^- \, d\nu \leq 1$. We will construct the appropriate measures $P$ and $Q$ over the extended real line $\overline{\mathbb{R}}$. Consequently, for $A \in \mathcal{B}(\overline{\mathbb{R}})$, let 
        \begin{align*}
            Q(A) &\coloneqq \nu(A\setminus\{-\infty\}),\\
            P(A) &\coloneqq \int_{A\setminus\{-\infty\}} \exp^-\,d\nu + \left(1-\int \exp^-\,d\nu\right)\mathbf{1}_A(-\infty).
        \end{align*}
        Define $\lambda(A) \coloneqq \frac{1}{2}\nu(A) + \frac{1}{2}\mathbf{1}_{A}(-\infty)$. 
        Thus, $P,Q \ll \lambda$. Furthermore, $\lambda \;- $ a.s.,
        \begin{align}
            q \coloneqq \frac{dQ}{d\lambda} &= 2 \mathbf{1}_{\mathbb{R}_{+\infty}}; \label{eq:lem-equiv-def-q-lambda-1}\\
            p \coloneqq \frac{dP}{d\lambda} &= 2 \exp^-\mathbf{1}_{\mathbb{R}_{+\infty}} + 2\left(1-\int \exp^-\,d\nu\right)\mathbf{1}_{\{-\infty\}},\label{eq:lem-equiv-def-q-lambda-2}
        \end{align}
        because $\int_A q \, d\lambda = Q(A)$ and $\int_A p \, d\lambda = P(A)$, for any $A$ in $\mathcal{B}(\overline{\mathbb{R}})$. Furthermore, $\log \frac{q}{p}: \overline{\mathbb{R}} \to \overline{\mathbb{R}}$ is the identity function on $\mathbb{R}_{+\infty}$. Hence, for any $A \in \mathcal{B}(\mathbb{R}\cup \{\infty\})$,
        \[
        Q\left(\log \frac{q}{p} \in A\right) = Q(A)= \nu(A).
        \]
        In other words, 
        $
        \nu = \left.Q\circ \left(\log \frac{q}{p} \right)^{-1}\right|_{\mathcal{B}(\mathbb{R}_{+\infty})}.
        $
    \end{IEEEproof}
\end{lemma}

Recall that we wish to study the behavior of $D_f: \mathscr{L} \to \mathbb{R}_{+\infty}$ under mixtures and convolutions of measures. We first note that $\mathscr{L}$ is closed
under mixtures.
\begin{lemma}
\label{lem:L-closed-under-operations}
    $\mathscr{L}$ is a convex set.
    \begin{IEEEproof}
        Let $\mu_1, \mu_2 \in \mathscr{L}$ and $\lambda \in [0,1]$. We have 
        \begin{align*}
            \int \exp^-\,d(\lambda \mu_1 + (1-\lambda)\mu_2) 
            &= \lambda\int \exp^-\,d\mu_1 + (1-\lambda)\int \exp^-\,d\mu_2 \leq 1.
        \end{align*}
        Hence, $\lambda \mu_1 + (1-\lambda)\mu_2$ is in $\mathscr{L}$. 
    \end{IEEEproof}
\end{lemma}

Similarly, using \eqref{eq:ineq-LLR-def} it is straightforward to verify that $\mathscr{L}$ is closed under the convolution operation. However, we provide a constructive argument for this fact, which details the connection between product measures and convolutions of LLR distributions. Under product measures, Radon-Nikodym derivatives factor into marginal densities, thus the LLR of a product measure against another is simply the sum of the marginal log-likelihood ratios:
\begin{lemma}
    We have 
    \begin{equation}
    \label{eq:LLR-dist-product-convolution}
    \hat{\mu}(P_1 \otimes P_2 \Vert Q_1 \otimes Q_2) = \hat{\mu}(P_1 \Vert Q_1) * \hat{\mu}(P_2 \Vert Q_2),
    \end{equation}
    for any pairs of probability measures $(P_1, Q_1)$ and $(P_2, Q_2)$ with each component defined on the same measurable space. 
    \begin{IEEEproof}
    Using the shorthand $p_i \coloneqq \frac{dP_i}{d\lambda_i}$ 
    and $q_i \coloneqq \frac{dQ_i}{d\lambda_i}$:
    \begin{equation}
    \label{eq:RN-product-rule}
    \log \frac{d(Q_1 \otimes Q_2)/ d(\lambda_1 \otimes \lambda_2)}{d(P_1 \otimes P_2)/d(\lambda_1 \otimes \lambda_2)}
    = \log \frac{q_1}{p_1} + \log \frac{q_2}{p_2}, \quad Q_1 \otimes Q_2 - \text{a.s.},
    \end{equation}
    where the sum on the right-hand side is interpreted as the function on 
    the product sample space: $(\omega_1,\omega_2) \mapsto \log \frac{q_1}{p_1}(\omega_1) + \log \frac{q_2}{p_2}(\omega_2)$. By Definition~\ref{def:convolution}, $\hat{\mu}(P_1 \Vert Q_1) * \hat{\mu}(P_2 \Vert Q_2)$ is exactly the push-forward of $Q_1 \otimes Q_2$ under such map.
    \end{IEEEproof}
\end{lemma}

Notice that Lemma~\ref{lemma:Df-df} follows immediately from Equation~\eqref{eq:LLR-dist-product-convolution}. 

\section{On the Formulation of Tensorization}
\label{appendix:domain-Df}
While $f$-divergences and their tensorization are commonly formulated over the proper class of all measurable spaces, in practice one may work within subclasses of measurable spaces that are closed under the product (e.g., finite/discrete spaces or Euclidean spaces with their Borel $\sigma$-algebras). 
Our main technical contributions (Sections~\ref{sec:tensorization-bases-are-affine}, \ref{sec:local-characteristic} and \ref{sec:main-characterization}) extend to any such subclass, since our proofs require only tensorization for finite-support discrete distributions, which always exist due to the closure assumption, unless the subclass contains no nontrivial measurable spaces. 

Admittedly, if $\mathcal{C}$ denotes a family of measurable spaces that is closed under products, and we take $\mathscr{L}(\mathcal{C})$ as the set of all LLR distributions induced by pairs of probability measures on some member of $\mathcal{C}$, then $\mathscr{L}(\mathcal{C}) \subseteq \mathscr{L}$, but it is not generally true that $\mathscr{L}(\mathcal{C}) = \mathscr{L}$; not every possible LLR distribution can be generated by considering two measures on an element of $\mathcal{C}$. However, this is not an issue as Lemma~\ref{lem:L-closed-under-operations} still holds for $\mathscr{L}(\mathcal{C})$.
\begin{lemma}
    Let $\mathcal{C}$ be a set of measurable spaces that is closed under product; that is, 
    \[\text{if } (E, \mathcal{E}), (F, \mathcal{F}) \in \mathcal{C}, \text{then } (E \times F, \mathcal{E} \otimes \mathcal{F}) \in \mathcal{C}.\]
    Let 
    \[\mathscr{L}(\mathcal{C}) = \{\hat{\mu}(P \Vert Q) \mid \text{$P$ and $Q$ are probability measures on $\mathcal{X}\in \mathcal{C}$} \}.\]
\end{lemma}
Then, $\mathscr{L}(\mathcal{C})$ is closed under mixture and convolution. 

Thus, as long as $\mathcal{C}$ contains a measurable space with at least two events, we can induce every two-point LLR distribution using elements of $\mathcal{C}$, which is enough to span the whole range of $D_f$, i.e.,  
$
\mathcal{R}(D_f) = \{D_f(\mu) : \mu \in \mathscr{L}(\mathcal{C})\}.
$

\section{Existence of the Local Characteristic and Uniqueness of Symmetrized Tensorizable $f$-divergences}
This appendix proves the two technical results underlying Section~\ref{sec:local-characteristic}: existence of the local characteristic and characterization of the restriction $g$ from the local characteristic. The core mechanism is the ``doubling relation'' in Equation~\eqref{eq:recursive_relation}. This relation implies that the map $h: r \mapsto g(r)/r^2$ is approximately invariant under successive halvings. This yields convergence of $h(r)$ as $r\to 0^+$, and later implies that the restrictions of symmetrized $f$-divergences that share a local characteristic and a tensorization base must agree everywhere.

\subsection{Existence of the Local Characteristic}
\label{appendix:proof-local-characteristic}
Throughout this subsection, $g(r)$ denotes the evaluation of a fixed symmetrized tensorizable $f$-divergence $D^{\text{sym}}_f$ at the symmetric LLR distribution with support $\{-r,r\}$, where $g(r) = 0$ if $r = 0$. Next, we record a few results that ensure the existence of the limit 
\[
\lim_{r \to 0^+} \frac{g(r)}{r^2}. 
\]
The main idea is to show that along a sequence of dyadic numbers, the function $h:r\mapsto g(r)/r^2$ induces a Cauchy sequence, and then prove that $h$ converges to the limit of this sequence as $r \to 0^+$. For this purpose, we derive a recursive equation for $g$ along ``dyadic (halving) steps.''
\begin{lemma}
\label{lem:recursive_relation_g}
    Define \[T(r) \coloneqq \left(\frac{1}{1+\tanh^2 \left(\frac{r}{2}\right)}\right).\]
    There exists $\gamma \in \mathbb{R}$ such that
    \begin{equation}
        \label{eq:recursive_relation}
        g(2r) = 2T(r)\left[2g(r) + \gamma g^2(r)\right], \quad r \in [0, \infty).
    \end{equation}
    \begin{IEEEproof}
        For any $r \in [0,\infty)$, simply take $r_1 = r = r_2$ in \eqref{eq:key-equation}.
    \end{IEEEproof}
\end{lemma}

Our first milestone is to prove that the sequence $\{h(2^{-k}x)\}$ converges as $k \to \infty$ for all $x \in (0, \infty)$. We start by showing that $g = O(r) \,\; (r \to 0^+)$ by invoking the recursive equation in \eqref{eq:recursive_relation} and the fact that a convex function is locally Lipschitz; this will eventually ensure that $h$ is bounded near the origin. 

\begin{lemma}
\label{lem:g_is_O(r)}
    There exist positive numbers $L$ and $r_0$ such that
    \begin{equation}
    \label{eq:g_is_O(r)}
        g(r) \leq 4L\,r, \quad  \forall r \in [0,r_0]. 
    \end{equation}
    \begin{IEEEproof}
         Since $f$ is convex, it is locally Lipschitz around $1$; using that $f(1)=0$, there exist $L>0$ and $\delta\in(0,1)$ such that
        \[
        |f(x)|\leq L|x-1|,\quad x\in(1-\delta,1+\delta).
        \]
        If $r \in [0,\log 2]$, then $|e^{\pm r}-1|\leq 2r$. Thus, if we let $r_0\coloneqq \min\{\log(1+\delta),-\log(1-\delta)\}$, then $|f(e^{\pm r})|\leq 2Lr$ for $r \in [0,r_0]$. Note that $r \in [0,r_0]$ implies $r \in [0, \log 2)$. Thus, 
        \begin{align*}
            g(r) &= f(e^{-r}) \cdot \frac{1}{1 + e^{-r}} + f(e^{r}) \cdot \frac{e^{-r}}{1 + e^{-r}}  \\
            &\leq \lvert f(e^{-r}) \rvert + \lvert f(e^{r}) \rvert \\
            &\leq 4rL.        
            \end{align*}
    \end{IEEEproof}
\end{lemma}

\begin{lemma}
\label{lem:Cauchy-seq}
    Let $x \in (0, \infty)$. Define two sequences $\{r_k(x)\}$ and $\{h_k(x)\}$ via
    \begin{equation}
        r_k(x) \coloneqq 2^{-k}x \quad \text{and} \quad h_k(x) \coloneqq  h(r_k(x)), \quad k \in \mathbb{Z}_{\geq 0}.
    \end{equation}
    Either there exists $K \in \mathbb{Z}_{\geq 0}$ such that $\{\log h_k(x)\}_{k \geq K}$ is Cauchy, or $h_k(x) \to 0$.
    \begin{IEEEproof}
        First, note that Equation~\eqref{eq:recursive_relation} implies
        \begin{equation}
        \label{eq:h2r-hr-factorized}
            h(2r) = T(r)\,h(r)\left(1+\frac{\gamma}{2}g(r)\right), \quad \forall r \in (0, \infty). 
        \end{equation}
        Then, because $g(r) \to 0$ as $r\to 0^+$, there exists $\delta > 0$  such that $\lvert \gamma g(r) \rvert < 1$ for all $0 < r < \delta$, which ensures that $1+\frac{\gamma}{2}g(r) > 0$ for $0 < r < \delta$. Define $r^* \coloneqq \min(r_0, \delta)$, where $r_0$ is given by Lemma~\ref{lem:g_is_O(r)}. For a fixed $x \in (0,\infty)$, choose $K$ large enough such that $r_K(x) < r^*$. 
        By all the preceding, for every $k \geq K$
        \[
            h_{k-1}(x)= T(r_k(x))h_k(x)\left(1+\frac{\gamma}{2}g(r_k(x))\right),
        \]
        and $1+\frac{\gamma}{2}g(r_k(x)) > 0$. Consequently, if $h_{k^*}(x) = 0$ for some $k^* \geq K$, then $h_k(x) = 0$ for all $k \geq K$, in which case $h_k(x) \to 0$. Hence, suppose that $h_k(x) > 0$ for every $k \geq K$. Because $T$ is a positive function and $1+\frac{\gamma}{2}g(r_k(x)) > 0$ for all $k \geq K$, the quantity $\log h_k(x)$ is well-defined for every $k \geq K$ and Equation~\eqref{eq:h2r-hr-factorized} yields
        \begin{equation}
            |\log h_{k}(x) - \log h_{k-1}(x)| = \left|\log T(r_k(x)) + \log \left(1+\frac{\gamma}{2}g(r_k(x))\right)\right|
        \end{equation}
    Given that $0 < T \leq 1$, for all $k \in \{1, 2, \dots\}$:
    \begin{align*}
        \lvert \log\,T(r_k(x))\rvert &= \log \left(1+\tanh^2 \left(\frac{2^{-k}x}{2}\right)\right) \\
        &\leq 4^{-k-1}x^2,
    \end{align*}
    where this inequality follows from the elementary inequalities: 
    \[
    \log(x+1) \leq x \quad \text{and} \quad 0 \leq \tanh(x) \leq x,\quad \forall x \in (0, \infty).
    \]
    Recall that $r_K(x) < r_0$. Thus, invoking Lemma~\ref{lem:g_is_O(r)} and the fact that $\lvert \log (1 + x) \rvert \leq 2 \vert x \vert$ for all $ x \in (-\frac{1}{2}, \frac{1}{2})$, we obtain
    \begin{align*}
        \left\lvert \log \left(1+\frac{\gamma}{2}g(r_k(x))\right) \right\rvert &\leq \lvert \gamma \rvert g(r_k(x)) \\
        &\leq 4 \lvert\gamma\rvert L 2^{-k}x, \quad \forall k \geq K. 
    \end{align*}
    Hence, for every $m \geq K + 1$,
    \begin{equation}
    \label{eq:series-to-zero}
        0 \leq \sum_{k = m}^\infty \lvert \log\,T(r_k(x))\rvert \leq \frac{4^{-m}x^2}{3} \quad \text{and} \quad 0 \leq \sum_{k = m}^\infty \left\lvert \log \left(1+\frac{\gamma}{2}g(r_k(x))\right) \right\rvert \leq \lvert \gamma \rvert L 2^{-m+3}x.
    \end{equation}
    The triangle inequality implies
    \[
    \sum_{k = m}^\infty |\log h_{k}(x) - \log h_{k-1}(x)| \xrightarrow[]{m \to \infty} 0.
    \]
    Therefore, $\{\log h_k(x)\}_{k \geq K}$ is a Cauchy sequence and converges to some finite value.  
    \end{IEEEproof} 
\end{lemma}

By Lemma~\ref{lem:Cauchy-seq} and continuity of the exponential function, the sequence $\{h_k(x)\}$ converges for every $x \in (0 ,\infty)$. Thus, we denote 
\begin{equation}
    \lim_{k \to \infty} h_k(x) = \ell_x,\quad \forall x \in (0, \infty).
\end{equation}
From the proof of Lemma~\ref{lem:Cauchy-seq} we extract a uniform bound on the sequence $\{\log h_k(x)\}$ across all initial points in a specific open interval. 
\begin{lemma}
\label{lem:log-bound}
    There exist $r^* > 0$ and $C \in \mathbb{R}$ such that either $h(x)=\ell_x=0$ or 
    \begin{equation}
    \label{eq:log-bound}
        \left|\log h(x)-\log\ell_x\right| \leq  Cx, 
    \end{equation}
     for all $ x \in (0, r^*)$.
    \begin{IEEEproof}
    Let $r^*$ be defined as in the proof of Lemma~\ref{lem:Cauchy-seq}. Take $x \in (0, r^*)$ and following the same arguments as in the proof of Lemma~\ref{lem:Cauchy-seq}, we obtain
    \begin{equation}
    \label{eq:log-difference}
        \lvert  \log h(x) - \log h_k(x) \rvert \leq \frac{x^2}{12}  + 4 \lvert \gamma \rvert L \,x \leq \left(\frac{r^*}{12}  + 4 \lvert \gamma \rvert L \right)x,
    \end{equation}
    where $L$ is the constant given by Lemma~\ref{lem:g_is_O(r)}.
       The proof concludes by taking the limit of $k\rightarrow\infty$ and defining $C \coloneqq \frac{r^*}{12}+4\lvert\gamma\rvert L $.
    \end{IEEEproof}
\end{lemma}
\begin{remark}
By iterating the refinement argument used above, the right-hand side of \eqref{lem:log-bound} can be improved to $Cx^2$. We omit the details since this sharper bound is not needed for our proof.
\end{remark}
Our next step is to show that for every $x > 0$, the value $\ell_x$ agrees with every other limit $\ell_y$ whenever $y$ lies on the dense set \[
\mathcal{D}(x) \coloneqq \{nx/2^k : n\in \mathbb{Z}_{>0},\; k \in \mathbb{Z}_{\geq 0}\} \subseteq (0, \infty),\]
thereby establishing a common convergence limit, together with a uniform bound, a dense subset of $(0, \infty)$. 
We prove this fact by extending the relation
\begin{equation}
    \label{eq:same-limits-halving}
    \ell_x = \ell_{y}, \quad \forall  y \in \{x/2^n : n \in \mathbb{Z}\},
\end{equation}
with the aid of the following lemma. Note that \eqref{eq:same-limits-halving} follows immediately by noting that, for every $n \in \mathbb{Z}$, the sequences $\{h_k(x/2^n)\}$ and $\{h_k(x)\}$ agree after shifting indices appropriately. 
\begin{lemma}
\label{lem:quadratic-equation-limits}
    Let $r_1, r_2 \in (0, \infty)$ such that $r_1 > r_2$. Then, 
    \begin{equation}
    \label{eq:quadratic-limits}
        (r_1 + r_2)^2 \ell_{r_1 + r_2} + (r_1 - r_2)^2 \ell_{r_1 - r_2} = 2r_1^2\ell_{r_1} + 2r_2^2\ell_{r_2}. 
    \end{equation}
    \begin{IEEEproof}
         To ease notation, define 
        \begin{equation}
            c_+(r_1, r_2) = \frac{1+e^{-r_1-r_2}}{(1+e^{-r_1})(1+e^{-r_2})}, \quad c_{-}(r_1, r_2) = \frac{1+e^{r_2-r_1}}{e^{r_2}(1+e^{-r_1})(1+e^{-r_2})}.
        \end{equation}
    Clearly, both functions are continuous at $(0,0)$ and $c_{+}(0,0) = \frac{1}{2} = c_{-}(0,0)$. Take $k \in \mathbb{Z}_{>0}$ and $r_1, r_2 \in (0, \infty)$ such that $r_1 > r_2$. By Lemma~\ref{lem:g-r1+r2}, 
    \begin{multline}
    \label{eq:b4-quadratic}
        c_{+}(2^{-k}r_1, 2^{-k}r_2)g(2^{-k}(r_1 + r_2)) + c_{-}(2^{-k}r_1, 2^{-k}r_2)g(2^{-k}(r_1 - r_2)) \\ =g(2^{-k}r_1) + g(2^{-k}r_2) +\gamma g(2^{-k}r_1)g(2^{-k}r_2)
    \end{multline}
    Multiplying both sides of Equation~\eqref{eq:b4-quadratic} by $4^{k}$ and taking the limit as $k \to \infty$ yields \eqref{eq:quadratic-limits}.
    \end{IEEEproof}
\end{lemma}

\begin{lemma}
\label{lem:same-limits-rationals}
    Let $x \in (0, \infty)$. If
    \[\mathcal{D}(x) \coloneqq \{nx/2^k : n\in \mathbb{Z}_{>0},\; k \in \mathbb{Z}_{\geq 0}\},\]
    then
    \begin{equation}
        \label{eq:same-limits-dense}
        \ell_x = \ell_{y}, \quad \forall y \in \mathcal{D}(x).
    \end{equation}
    \begin{IEEEproof}
        Rewrite~\eqref{eq:same-limits-dense} as 
        \begin{equation}
        \label{eq:same-limits-rationals}
            \ell_x = \ell_{nx/2^k}, \quad \forall  n\in \mathbb{Z}_{>0},\; k \in \mathbb{Z}_{\geq 0}. 
        \end{equation}
        Our proof is by induction on $n$ for a fixed $k \in \mathbb{Z}_{\geq 0}$.
        Indeed, \eqref{eq:same-limits-rationals} holds for $n = 1$ due to \eqref{eq:same-limits-halving}. The case $n = 2$ follows from \eqref{eq:same-limits-halving} as well. Thus, assume \eqref{eq:same-limits-rationals} holds for all $ n < N+1$, with $N > 1$. Denote $r_1 = Nx/2^{k}$ and $r_2 = x/2^{k}$. By the induction hypothesis, 
        \[
            \ell_{r_1} = \ell_{r_2} = \ell_{r_1 - r_2} = \ell_x.
        \]
        Combining these equalities with Lemma~\ref{lem:quadratic-equation-limits}, specifically Equation~\eqref{eq:quadratic-limits}, yields
        \[
            \ell_{r_1 + r_2} = \ell_x. 
        \]
    Since $r_1 + r_2 = (N+1)x/2^{k}$, the previous equality closes the induction.
    \end{IEEEproof}
\end{lemma}

\begin{IEEEproof}\textbf{(Theorem~\ref{thm:sym_c})}
    First, note that Lemma~\ref{lem:log-bound}, together with Lemma~\ref{lem:same-limits-rationals}, implies a common $\ell_x$ for any choice of $x\in(0,r^*)$ that the following bound holds for all $y \in \mathcal{D}(x) \cap (0, r^*)$.
    \begin{equation}
    \label{eq:quartic-approx}
     h(y)\in\left[e^{-Cy}\ell_x, e^{Cy}\ell_x\right].
    \end{equation}
    Since $\mathcal{D}(x)$ is a dense subset, the above bound extends to any $y\in (0,r^*)$ by continuity of $h$ on $(0,\infty)$. Then, the theorem directly follows from sandwiching.
\end{IEEEproof}

\subsection{Uniqueness from the Local Characteristic}
\label{appendix:proof-of-unique-g}
We now prove that a solution of \eqref{eq:key-equation} with local characteristic can be uniquely identified using this quadratic rate and its tensorization parameter. The argument has two steps. First, we show that two solutions to \eqref{eq:key-equation} with the same local characteristic must agree on some neighborhood of the origin. Second, the recursive identity in \eqref{eq:recursive_relation} propagates this equality to $(0, \infty)$. 

\begin{definition}[Neighborhood of commonality]
    Given two functions $g_1,g_2: [0,\infty) \to \mathbb{R}$, we define 
    \begin{equation}
    \label{eq:neighborhood_of_commonality}
        \bm{\varepsilon}(g_1, g_2) = \sup \{x \in (0,\infty): \forall r \in [0,x), g_1(r) = g_2(r)\}.
    \end{equation}
\end{definition}

\begin{lemma}
\label{lem:neighborhood}
    Set $\gamma\in \mathbb{R}$ and $\ell \in [0, \infty)$. Let $g_1, g_2: [0,\infty) \to [0,\infty)$ be two continuous solutions to \eqref{eq:key-equation} with this $\gamma$. Suppose also that $g_1$ and $g_2$ share the same local characteristic $\ell$. Then 
    \begin{equation}
        \bm{\varepsilon}(g_1, g_2) > 0. 
    \end{equation}
    \begin{IEEEproof}
    Because $g_1$ and $g_2$ fulfill \eqref{eq:key-equation} for $\gamma \in \mathbb{R}$, taking $r_1 = r/2 = r_2$ in \eqref{eq:key-equation} yields
    \begin{equation}
    \label{eq:recursive-rel-g_i}
        g_i(r) = 2T\left(\frac{r}{2}\right)\left[2 g_i\left(\frac{r}{2}\right) + \gamma g_i^2\left(\frac{r}{2}\right)\right]; \quad \forall r \in [0,\infty), i \in\{1,2\}.
    \end{equation}
    For $r \in [0, \infty)$, denote $r_k = 2^{-k}r$. From \eqref{eq:recursive-rel-g_i}, for any integer $n \geq 1$, 
    \begin{equation}
    \label{eq:diff-g_1-g_2}
        \left\lvert \frac{g_1(r)}{r^2} - \frac{g_2(r)}{r^2} \right\rvert = \prod_{k = 1}^n\left\lvert T(r_k)\left[1 + \frac{\gamma}{2}(g_1(r_k) + g_2(r_k))\right]\right\rvert\; \left\lvert \frac{g_1(r_n)}{r_n^2} - \frac{g_2(r_n)}{r_n^2}\right\rvert.
    \end{equation}
    Given that $0 \leq T \leq 1$, then applying the triangle inequality to the right-hand side of \eqref{eq:diff-g_1-g_2} yields
    \begin{equation}
        \label{eq:ineq-g_1-g_2}
        \left\lvert \frac{g_1(r)}{r^2} - \frac{g_2(r)}{r^2} \right\rvert \leq \prod_{k = 1}^n \left[1 + \frac{\lvert \gamma \rvert}{2}(g_1(r_k) + g_2(r_k))\right]\; \left\lvert \frac{g_1(r_n)}{r_n^2} - \frac{g_2(r_n)}{r_n^2}\right\rvert, \quad \forall n \in \mathbb{Z}_{>0}.
    \end{equation}
    We argue that the right-hand side of \eqref{eq:ineq-g_1-g_2} can be made arbitrarily small. Given that
    \[
    \lim_{r \to 0^+}\frac{g_i(r)}{r^2} < \infty, \quad i \in \{1,2\},
    \]
    there exists $M >0$ and $r^* \in (0,\infty)$ such that $\vert g_i(r) \vert < r^2M$ for all $r < r^*$ and $i \in \{1,2\}$. The standard logarithmic inequality $\log(1+x) \leq x$ for $x > -1$, implies that the sequence 
    \[
    \prod_{k = 1}^n \left[1 + \frac{\lvert \gamma \rvert}{2}(g_1(r_k) + g_2(r_k))\right]
    \]
    is bounded from above whenever $r < r^*$. Moreover, because $g_1$ and $g_2$ share the same local characteristic, then 
    \begin{equation}
        \left\lvert \frac{g_1(r_n)}{r_n^2} - \frac{g_2(r_n)}{r_n^2}\right\rvert
    \end{equation}
    is arbitrarily small for large enough $n$. Hence, for any $\delta >0$, $\left\lvert \frac{g_1(r)-g_2(r)}{r^2}\right\rvert < \delta$ for every $r \in (0, r^*)$. Thus, $g_1 = g_2$ on the interval $(0, r^*)$. Since both function are continuous, we have that $g_1 = g_2$ on $[0,r^*)$. Consequently, $\bm{\varepsilon}(g_1, g_2) \geq r^* > 0$.
\end{IEEEproof}
\end{lemma}
\begin{IEEEproof}\textbf{(Theorem~\ref{thm:unique-g})}\,
    By Lemma~\ref{lem:neighborhood}, there exists $r^* \in (0, \infty)$ such that $g_1(r) = g_2(r)$ for any $r \in [0, r^*)$. Because $g_1$ and $g_2$ fulfill \eqref{eq:key-equation}, then we have 
    \[
    g_i(2r) = 2T(r)\left[2g_i(r) + \gamma g_i^2(r)\right], \quad r \in [0, \infty). 
    \]
    This relation implies that, in fact, $g_1(r) = g_2(r)$ for any $r \in [0,2^nr^*)$ and any $n \in \mathbb{Z}_{>0}$. Thus, $g_1 = g_2$. 
\end{IEEEproof}

\section{Solution Space of the Semi-Symmetrized Tensorization Equation and Characterization of Tensorizable $f$-Divergences}
\label{appendix:sols-characterization-semi-sym}
This appendix presents the technical details of Section~\ref{sec:main-characterization}, including the proofs of Theorems~\ref{thm:basis-kernel} and~\ref{thm:characterization-tensorization}. 

\begin{lemma}
    Let $f: (0, \infty) \to \mathbb{R}$, where $f(1) = 0$, be a continuous solution to Equation~\eqref{eq:another-key-equation} for a fixed $g: [0,\infty) \to [0, \infty)$ and $\gamma \in \mathbb{R}$. Then, 
    \begin{equation}
        \label{eq:consistency}
        T_r[f](1) = g(r), \quad \forall r \in [0, \infty). 
    \end{equation}
    Furthermore, for every $h \in \ker \phi_g^\gamma$ 
    \begin{equation}
    \label{eq:consistency-kernel}
        T_r[h](1) = 0, \quad \forall r \in [0,\infty).
    \end{equation}
    \begin{IEEEproof}
        If $f$ solves \eqref{eq:another-key-equation}, then $\phi_g^\gamma[f](s_1, s_2, r) = g(r)$ for all $s_1 < 0 < s_2$ and $r \geq 0$. By continuity of $f$ and the condition $f(1) = 0$, taking the limit as $s_1 \to 0^-$ yields 
        \[
        T_r[f](1) - (1+\gamma g(r))f(1) = T_r[f](1) = g(r), \quad r \in [0, \infty).
        \]
        The same argument yields $T_r[h](1) = 0$ for $r \geq 0$ and $h \in \ker \phi_g^\gamma$.
    \end{IEEEproof}
\end{lemma}

\begin{lemma}
\label{lem:dimension-log-grid}
    Fix $r > 0$ and define 
    \[
    e^{r\mathbb{Z}} \coloneqq \{e^{nr} : n \in \mathbb{Z}\} \quad \text{and} \quad
    \mathcal{K}_r \coloneqq \{h |_{e^{r\mathbb{Z}}}: h \in \ker \phi_g^\gamma\}.
    \]
    Then, $\dim \mathcal{K}_r \leq 2$. 
    \begin{IEEEproof}
        Let $h \in \ker \phi_g^\gamma$. We will inductively show that $h|_{e^{r\mathbb{Z}}}$ is completely determined by the two values $h(e^r)$ and $h(e^{2r})$. The initial cases are handled by \eqref{eq:consistency-kernel}. Specifically, from \eqref{eq:consistency-kernel}, we get
        \begin{equation}
            \label{eq:consistency-pos-neg}
             h(e^{-nr}) = -e^{-nr}h(e^{nr}), \quad \forall n \in\mathbb{Z}_{>0}.
        \end{equation}
        Hence, $h(e^{-r})$ and $h(e^{-2r})$ are determined by $h(e^r)$ and $h(e^{2r})$, respectively. Let $n \in \mathbb{Z}_{>0}$ and set $(s_1, s_2) = (-r,nr)$. Because $h \in \ker \phi_g^\gamma$:
        \begin{equation}
        \label{eq:inductive-log-kernel}
            \rho(-r) \left[\frac{h(e^{-(n+1)r}) + e^{-r}h(e^{-(n-1)r})}{1+e^{-r}}-(1+\gamma g(r))h(e^{-nr})\right] - \rho(nr)\left[\frac{e^{-r}h(e^{2r})}{1+e^{-r}}-(1+\gamma g(r))h(e^r)\right] = 0.
        \end{equation}
        Since $\rho(-r) \neq 0$, setting $n = 2$ in \eqref{eq:inductive-log-kernel} proves that $h(e^{-3r})$ is also determined by the tuple $(h(e^r), h(e^{2r}))$. Proceeding inductively, we get $h(e^{-nr})$ is recoverable from $(h(e^r), h(e^{2r}))$ for every $n > 0$. Equation~\eqref{eq:consistency-pos-neg} then extends this property to $h(e^{nr})$ for every $n > 0$. Finally, because $h(e^0) = 0$, we conclude that $h|_{e^{r\mathbb{Z}}}$ is determined by the two values $h(e^r)$ and $h(e^{2r})$. Hence, $\dim \mathcal{K}_r \leq 2$.
    \end{IEEEproof}
\end{lemma}

\begin{IEEEproof}\textbf{(Theorem~\ref{thm:basis-kernel})}
Let $u_1, u_2, u_3 \in \ker \phi_g^\gamma$ and $a,b,c \in \mathbb{Z}$. Fix $r > 0$, then by Lemma~\ref{lem:dimension-log-grid} we have 
\[
F(e^{ar}, e^{br}, e^{cr}) \coloneqq \det \begin{pmatrix}
    u_1(e^{ar}) & u_2(e^{ar}) & u_3(e^{ar}) \\
    u_1(e^{br}) & u_2(e^{br}) & u_3(e^{br}) \\
    u_1(e^{cr}) & u_2(e^{cr}) & u_3(e^{cr}) \\
\end{pmatrix} = 0. 
\]
The set $\{(e^{ar}, e^{br}, e^{cr}): r > 0; a,b,c \in \mathbb{Z}\}$ is dense in $(0,\infty)^3$ and $F: (0,\infty)^3 \to \mathbb{R}$ is continuous. Hence, $F(x,y,z) = 0$ for all $x,y,z \in (0,\infty)$. It follows that $\{u_1, u_2, u_3\}$ is linearly dependent. Since $u_1, u_2$ and $u_3$ were arbitrary elements of the kernel, we have $\dim \ker \phi_g^\gamma \leq 2$. 

We now prove that each basis listed in \eqref{eq:bases} is indeed a basis of $\ker \phi_g^\gamma$. Clearly, each set in \eqref{eq:bases} is linearly independent, so it remains only to verify that its elements belong to the kernel in the corresponding regime of $\gamma$ and $\ell$. First, notice that the map $h_0: x \mapsto x - 1$ is always in $\ker \phi_g^\gamma$ for every pair of valid $g$ and $\gamma \in \mathbb{R}$. Indeed, it is straightforward to check that $T_r[h_0] = h_0$. Combining this with the fact that $h_0(e^{-s}) = \rho(s)$, shows that $\phi_g^\gamma[h_0] = 0$ for any valid $g$ and $\gamma$. Now, suppose $\gamma \ell = 0$ and define $h_1: x \mapsto (x+1)\log x$. Corollary~\ref{cor:explicit-form-g} implies that $\gamma g = 0$. This fact, along with the equality 
\[
T_r[h_1](e^{-s})-h_1(e^{-s}) = \frac{r(1-e^{-r})}{1+e^{-r}}\rho(s), \quad \forall s \in \mathbb{R},
\]
yields that $\phi_g^\gamma[h_1] = 0$ if $\gamma \ell = 0$. Next, assume $\gamma\ell > -\frac{1}{8}$ and $\gamma \ell \neq 0$. Let $\alpha$ be any root of the polynomial $x(x-1) - 2\gamma \ell$ and let $h_2: x \mapsto x^\alpha - x^{1-\alpha}$. In this case, we have 
\[
T_r[h_2] = \frac{e^{-\alpha r}+e^{-(1-\alpha)r}} {1+e^{-r}}h_2 \quad \text{and}\quad 1+\gamma g(r) = \frac{e^{-\alpha r}+e^{-(1-\alpha)r}}
{1+e^{-r}}.
\]
Hence, $T_r[h_2] - (1+\gamma g(r))h_2 = 0$, and $\phi_g^\gamma[h_2] = 0$.  Finally, if $\gamma \ell = -\frac{1}{8}$, define $h_3: x \mapsto \sqrt{x}\log x$. Then, $T_r[h_3] = \frac{h_3}{\cosh(r/2)}$. By \eqref{eq:explicit-form-g}, if $\gamma\ell = -\frac{1}{8}$, we also have that $1 + \gamma g(r) = \frac{1}{\cosh(r/2)}$. Hence, $T_r[h_3] - (1+\gamma g(r))h_3 = 0$, and $h_3 \in \ker \phi_g^\gamma$.
\end{IEEEproof}

\begin{IEEEproof}\textbf{(Theorem~\ref{thm:characterization-tensorization})}
    Note that any solution to \eqref{eq:another-key-equation} is of the form $f = f_0 + h$, where $f_0$ is a particular solution and $h \in \ker \phi_g^\gamma$. We leverage this observation and Theorem~\ref{thm:basis-kernel} to derive an explicit expression for $f$. By Theorem~\ref{thm:relevant-ell-gamma-case}, we know $\ell \gamma \geq -\frac{1}{8}$. Thus, suppose $\ell = 0$. Corollary~\ref{cor:explicit-form-g} implies $g = 0$. Because 
    \[
    g(r) = \frac{1}{1+e^{-r}}f(e^{-r}) + \frac{e^{-r}}{1+e^{-r}}f(e^{r}) = 0 = f(1), \quad \forall r \in [0,\infty),
    \]
    the convex function $f$ is an affine map on $[e^{-r}, e^{r}]$, for all $r >0$. It follows that, for some $E \in \mathbb{R}$,
    \[
    f(x) = E(x-1), \quad \forall x \in (0,\infty). 
    \]
    Henceforth, assume $\ell > 0$. Now, if $\gamma = 0$, we take
    \[
    f_0(x)\coloneqq 2\ell x\log x,\qquad x\in(0,\infty),
    \] which is the generator of $2\ell$ times the KL divergence. By Theorem~\ref{thm:basis-kernel}, every solution $f$ to \eqref{eq:another-key-equation} is of the form
    \[
    f(x)=Ax\log x- B\log x + E(x-1), \quad x \in (0, \infty),
    \]
    where $A,B,E \in \mathbb{R}$ and $A + B = 2\ell$. Convexity then enforces $A, B \geq 0$ (this fact can be easily checked using the second derivative of $f$). Next, if $\gamma \neq 0$, we proceed in two parts. On the one hand, if $\gamma \ell = - \frac{1}{8}$, then $\alpha = \frac{1}{2}$. Take as the base solution to \eqref{eq:another-key-equation}
    \begin{equation}
        f_0(x) = \frac{x^{1/2}-1}{\gamma}, \quad x \in (0,\infty), 
    \end{equation}
    which has local characteristic $\ell$. Invoking Theorem~\ref{thm:basis-kernel}, we have that 
    \[
    f(x) =  \frac{x^{1/2}-1}{\gamma} + a\,\sqrt{x} \log x + E\,(x-1), \quad x \in (0,\infty),
    \]
    for some $a, E \in \mathbb{R}$. It is straightforward to check that $f$ can only be convex on $(0, \infty)$ if $a = 0$. On the other hand, if $\gamma\ell > -\frac{1}{8}$, we take $\alpha$ such that $\alpha(\alpha - 1) = 2\gamma \ell$. Let
    \[
    f_0(x) = \frac{x^\alpha - 1}{\gamma}, \quad x \in (0, \infty).
    \]
    Again, invoking Theorem~\ref{thm:basis-kernel}, we know there exists $A \in \mathbb{R}$ such that
    \[
    f(x) = A \frac{x^\alpha -1}{\gamma} + (1-A) \frac{x^{1-\alpha} - 1}{\gamma} + E(x-1), \quad x \in (0, \infty).
    \]
    If $A < 0$ or $A > 1$, it is always possible to find an interval where $f'' < 0$ since the coefficients $A$ and $1-A$ have different signs and are not zero. Thus, $A \in [0,1]$; furthermore, $A \in \{0,1\}$. To see this, let $D_{\alpha,1/\gamma}$ and $D_{1-\alpha, 1/\gamma}$ be the power divergences induced by $x \mapsto (x^\alpha -1)/\gamma$ and $x \mapsto (x^{1-\alpha} - 1)/\gamma$, respectively. Then, $D_f = A D_{\alpha,1/\gamma} + (1-A)D_{1-\alpha, 1/\gamma}$. Note that both $D_{\alpha,1/\gamma}$ and $D_{1-\alpha, 1/\gamma}$ have tensorization parameter $\gamma$. Combining this with the tensorization identity for $f$:
    \[
    D_f(P_1\otimes P_2 \Vert Q_1 \otimes Q_2) = D_f(P_1 \Vert Q_1) + D_f(P_2 \Vert Q_2) + \gamma D_f(P_1 \Vert Q_1)D_f(P_2 \Vert Q_2),
    \]
    yields
    \begin{equation}
    \label{eq:A-zero-or-one}
        A(1-A)(D_{\alpha,1/\gamma}(P_1 \Vert Q_1) - D_{1-\alpha, 1/\gamma}(P_1 \Vert Q_1))(D_{1-\alpha, 1/\gamma}(P_2 \Vert Q_2)-D_{\alpha,1/\gamma}(P_2 \Vert Q_2)) = 0
    \end{equation}
    for any probability measures $P_1, Q_1, P_2$ and $Q_2$ such that the involved values of the $f$-divergences above are finite. Because $\gamma \neq 0$, $\ell \neq 0$ and $\gamma \ell > -\frac{1}{8}$, we can always find measures $P$ and $Q$ such that $D_{\alpha,1/\gamma}(P \Vert Q) \neq D_{1-\alpha, 1/\gamma}(P \Vert Q)$ since $\alpha \notin \left\{0, 1, \frac{1}{2}\right\}$. In fact, for sufficiently small $\epsilon > 0$
    \[
        D_{\alpha,1/\gamma}(\operatorname{Bern}(1/2) \Vert \operatorname{Bern}(\epsilon)) \neq D_{1-\alpha, 1/\gamma}(\operatorname{Bern}(1/2) \Vert \operatorname{Bern}(\epsilon)), \quad \alpha \notin \left\{0, 1, \frac{1}{2}\right\}. 
    \]
    Hence, \eqref{eq:A-zero-or-one} forces $A \in \{0,1\}$. Note that if $\alpha$ satisfies $\alpha(\alpha-1) = 2\gamma \ell$, so does $1-\alpha$. Thus, without loss of generality we assume $A = 1$, yielding
    \[
    f(x) = \frac{x^\alpha - 1}{\gamma} + E\,(x-1), \quad x \in (0,\infty).
    \]
    Finally, using elementary properties of logarithms and exponents, together with linearity of the integral and the Tonelli--Fubini theorem, it is straightforward to verify that the functions $f$ derived above induce tensorizable $f$-divergences.
\end{IEEEproof}

\bibliographystyle{IEEEtran}
\bibliography{bibliography}

@book{kallenberg1997foundations,
  title={Foundations of modern probability},
  author={Kallenberg, Olav},
  year={2021},
  edition = {3},
  publisher={Springer}
}

@article{vajda1972f,
  title={On the f-divergence and singularity of probability measures},
  author={Vajda, Igor},
  journal={Periodica Mathematica Hungarica},
  volume={2},
  number={1-4},
  pages={223--234},
  year={1972}
}

@article{yu2023ising,
  title={Ising model on locally tree-like graphs: Uniqueness of solutions to cavity equations},
  author={Yu, Qian and Polyanskiy, Yury},
  journal={IEEE Transactions on Information Theory},
  volume={70},
  number={3},
  pages={1913--1938},
  year={2023},
  publisher={IEEE}
}

@article{gonccalves2024likelihood,
  title={Likelihood function: Definition, examples, and numerical experiments.},
  author={Gon{\c{c}}alves, Fl{\'a}vio B and Franklin, Pedro},
  journal={Chilean Journal of Statistics (ChJS)},
  volume={15},
  number={2},
  year={2024}
}

@article{LindleyD.V.1953SI,
author = {Lindley, D. V.},
issn = {0035-9246},
journal = {Journal of the Royal Statistical Society. Series B, Methodological},
number = {1},
pages = {30-76},
publisher = {Royal Statistical Society},
title = {Statistical Inference},
volume = {15},
year = {1953},
}

@article{HalmosPaulR.1949AotR,
author = {Halmos, Paul R. and Savage, L. J.},
journal = {The Annals of mathematical statistics},
number = {2},
pages = {225-241},
publisher = {Institute of Mathematical Statistics},
title = {Application of the Radon-Nikodym Theorem to the Theory of Sufficient Statistics},
volume = {20},
year = {1949},
}

@book{FollandG.B1984Ra,
publisher = {Wiley},
year = {1984},
title = {Real analysis: modern techniques and their applications},
address = {New York},
author = {Folland, Gerald B.},
}

@inproceedings{CruzRodrigo2025Tof,
pages = {1--6},
publisher = {IEEE},
booktitle = {IEEE International Symposium on Information Theory (ISIT)},
year = {2025},
title = {Tensorization of f-Divergences},
author = {Cruz, Rodrigo and Diaz, Mario and Calmon, Flavio P.}
}

@book{polyanskiy2025information,
  title={Information theory: From coding to learning},
  author={Polyanskiy, Yury and Wu, Yihong},
  year={2025},
  publisher={Cambridge university press}
}

@article{MuXiaosheng2021FBDi,
journal = {Econometrica},
pages = {475--506},
volume = {89},
publisher = {Wiley},
number = {1},
year = {2021},
title = {From Blackwell Dominance in Large Samples to Rényi Divergences and Back Again},
copyright = {2021 The Econometric Society},
author = {Mu, Xiaosheng and Pomatto, Luciano and Strack, Philipp and Tamuz, Omer}
}

@article{csiszar1967information,
  title={On information-type measure of difference of probability distributions and indirect observations},
  author={Csisz{\'a}r, Imre},
  journal={Studia Sci. Math. Hungar.},
  volume={2},
  pages={299--318},
  year={1967}
}

@inproceedings{balle2020hypothesis,
  title={Hypothesis testing interpretations and renyi differential privacy},
  author={Balle, Borja and Barthe, Gilles and Gaboardi, Marco and Hsu, Justin and Sato, Tetsuya},
  booktitle={International Conference on Artificial Intelligence and Statistics},
  pages={2496--2506},
  year={2020},
  organization={PMLR}
}

@inproceedings{abadi2016deep,
  title={Deep learning with differential privacy},
  author={Abadi, Martin and Chu, Andy and Goodfellow, Ian and McMahan, H Brendan and Mironov, Ilya and Talwar, Kunal and Zhang, Li},
  booktitle={Proceedings of the 2016 ACM SIGSAC conference on computer and communications security},
  pages={308--318},
  year={2016}
}

@inproceedings{mironov2017renyi,
  title={R{\'e}nyi differential privacy},
  author={Mironov, Ilya},
  booktitle={2017 IEEE 30th computer security foundations symposium (CSF)},
  pages={263--275},
  year={2017},
  organization={IEEE}
}

@inproceedings{mcsherry2009privacy,
  title={Privacy integrated queries: an extensible platform for privacy-preserving data analysis},
  author={McSherry, Frank D},
  booktitle={Proceedings of the 2009 ACM SIGMOD International Conference on Management of data},
  pages={19--30},
  year={2009}
}

@inproceedings{zhu2022optimal,
  title={Optimal accounting of differential privacy via characteristic function},
  author={Zhu, Yuqing and Dong, Jinshuo and Wang, Yu-Xiang},
  booktitle={International Conference on Artificial Intelligence and Statistics},
  pages={4782--4817},
  year={2022},
  organization={PMLR}
}

@article{mu2024monotone,
  title={Monotone additive statistics},
  author={Mu, Xiaosheng and Pomatto, Luciano and Strack, Philipp and Tamuz, Omer},
  journal={Econometrica},
  volume={92},
  number={4},
  pages={995--1031},
  year={2024},
  publisher={Wiley Online Library}
}

@article{mattner1999cumulants,
  title={What are cumulants?},
  author={Mattner, Lutz},
  journal={Documenta Mathematica},
  volume={4},
  pages={601--622},
  year={1999},
  publisher={Universi{\"a}t Bielefeld, Fakult{\"a}t f{\"u}r Mathematik Bielefeld}
}

@book{LeCamLucien1986AMiS,
series = {Springer Series in Statistics},
publisher = {Springer New York},
year = {1986},
title = {Asymptotic Methods in Statistical Decision Theory},
edition = {1},
author = {Le Cam, Lucien},
}

@book{TsallisConstantino2009ItNS,
publisher = {Springer New York},
year = {2009},
title = {Introduction to Nonextensive Statistical Mechanics : Approaching a Complex World},
edition = {1st ed.},
author = {Tsallis, Constantino}
}

@article{jizba2017uniqueness,
  title={On the uniqueness theorem for pseudo-additive entropies},
  author={Jizba, Petr and Korbel, Jan},
  journal={Entropy},
  volume={19},
  number={11},
  pages={605},
  year={2017},
  publisher={MDPI}
}

@article{leinster2019short,
  title={A short characterization of relative entropy},
  author={Leinster, Tom},
  journal={Journal of Mathematical Physics},
  volume={60},
  number={2},
  year={2019},
  publisher={AIP Publishing}
}

@article{furuichi2005uniqueness,
  title={On uniqueness theorems for Tsallis entropy and Tsallis relative entropy},
  author={Furuichi, Shigeru},
  journal={IEEE Transactions on Information Theory},
  volume={51},
  number={10},
  pages={3638--3645},
  year={2005},
  publisher={IEEE}
}

@article{suyari2004generalization,
  title={Generalization of Shannon-Khinchin axioms to nonextensive systems and the uniqueness theorem for the nonextensive entropy},
  author={Suyari, Hiroki},
  journal={IEEE Transactions on Information Theory},
  volume={50},
  number={8},
  pages={1783--1787},
  year={2004},
  publisher={IEEE}
}

@article{kakutani_equivalence_1948,
	title = {On {Equivalence} of {Infinite} {Product} {Measures}},
	volume = {49},
	language = {en},
	number = {1},
	journal = {The Annals of Mathematics},
	author = {Kakutani, Shizuo},
	year = {1948},
	pages = {214},
}

@article{yang_information-theoretic_1999,
	title = {Information-theoretic determination of minimax rates of convergence},
	volume = {27},
	number = {5},
	journal = {The Annals of Statistics},
	author = {Yang, Yuhong and Barron, Andrew},
	year = {1999},
}

@article{li2016renyi,
  title={R{\'e}nyi divergence variational inference},
  author={Li, Yingzhen and Turner, Richard E},
  journal={Advances in neural information processing systems},
  volume={29},
  year={2016}
}

@article{nowozin2016f,
  title={f-gan: Training generative neural samplers using variational divergence minimization},
  author={Nowozin, Sebastian and Cseke, Botond and Tomioka, Ryota},
  journal={Advances in neural information processing systems},
  volume={29},
  year={2016}
}

@book{kraft_conditions_1955,
	series = {University of {California} publications in statistics v. 2 no. 6},
	title = {Some conditions for consistency and uniform consistency of statistical procedures},
	publisher = {University of California Press},
	author = {Kraft, Charles H.},
	year = {1955},
}

@InProceedings{pmlr-v19-garivier11a,
  title = 	 {The KL-UCB Algorithm for Bounded Stochastic Bandits and Beyond},
  author = 	 {Garivier, Aurélien and Cappé, Olivier},
  booktitle = 	 {Proceedings of the 24th Annual Conference on Learning Theory},
  pages = 	 {359--376},
  year = 	 {2011},
  editor = 	 {Kakade, Sham M. and von Luxburg, Ulrike},
  volume = 	 {19},
  series = 	 {Proceedings of Machine Learning Research},
  address = 	 {Budapest, Hungary},
  month = 	 Jun,
  publisher =    {PMLR}
}

@article{10.1016/0196-8858,
author = {Lai, T.L and Robbins, Herbert},
title = {Asymptotically efficient adaptive allocation rules},
year = {1985},
issue_date = {March, 1985},
publisher = {Academic Press, Inc.},
address = {USA},
volume = {6},
number = {1},
issn = {0196-8858},
journal = {Adv. Appl. Math.},
month = mar,
pages = {4–22},
numpages = {19}
}

@inproceedings{10.5555/795662.796294,
author = {Auer, P. and Cesa-Bianchi, N. and Freund, Y. and Schapire, R. E.},
title = {Gambling in a rigged casino: The adversarial multi-armed bandit problem},
year = {1995},
isbn = {0818671831},
publisher = {IEEE Computer Society},
address = {USA},
booktitle = {Proceedings of the 36th Annual Symposium on Foundations of Computer Science},
pages = {322},
series = {FOCS '95}
}

@inproceedings{audibert:hal-00834882,
  TITLE = {{Minimax policies for adversarial and stochastic bandits}},
  AUTHOR = {Audibert, Jean-Yves and Bubeck, S{\'e}bastien},
  BOOKTITLE = {{Proceedings of the 22th annual conference on learning theory}},
  ADDRESS = {Montreal, Canada},
  PAGES = {217-226},
  YEAR = {2009},
  MONTH = Jun,
  HAL_ID = {hal-00834882},
  HAL_VERSION = {v1},
}

@ARTICLE{7055287,
  author={Duchi, John C. and Jordan, Michael I. and Wainwright, Martin J. and Wibisono, Andre},
  journal={IEEE Transactions on Information Theory}, 
  title={Optimal Rates for Zero-Order Convex Optimization: The Power of Two Function Evaluations}, 
  year={2015},
  volume={61},
  number={5},
  pages={2788-2806}
  }

@InProceedings{pmlr-v206-yu23a,
  title = 	 {Optimal Sample Complexity Bounds for Non-convex Optimization under Kurdyka-Lojasiewicz Condition},
  author =       {Yu, Qian and Wang, Yining and Huang, Baihe and Lei, Qi and Lee, Jason D.},
  booktitle = 	 {Proceedings of The 26th International Conference on Artificial Intelligence and Statistics},
  pages = 	 {6806--6821},
  year = 	 {2023},
  editor = 	 {Ruiz, Francisco and Dy, Jennifer and van de Meent, Jan-Willem},
  volume = 	 {206},
  series = 	 {Proceedings of Machine Learning Research},
  month = 	 Apr,
  publisher =    {PMLR}
}

@inproceedings{NEURIPS2024_b444ad72,
 author = {Yu, Qian and Wang, Yining and Huang, Baihe and Lei, Qi and Lee, Jason},
 booktitle = {Advances in Neural Information Processing Systems},
 editor = {A. Globerson and L. Mackey and D. Belgrave and A. Fan and U. Paquet and J. Tomczak and C. Zhang},
 pages = {99564--99600},
 publisher = {Curran Associates, Inc.},
 title = {Stochastic Zeroth-Order Optimization under Strongly Convexity and Lipschitz Hessian: Minimax Sample Complexity},
 volume = {37},
 year = {2024}
}

@InProceedings{pmlr-v195-foster23b,
  title = 	 {Tight Guarantees for Interactive Decision Making with the Decision-Estimation Coefficient},
  author =       {Foster, Dylan J. and Golowich, Noah and Han, Yanjun},
  booktitle = 	 {Proceedings of Thirty Sixth Conference on Learning Theory},
  pages = 	 {3969--4043},
  year = 	 {2023},
  editor = 	 {Neu, Gergely and Rosasco, Lorenzo},
  volume = 	 {195},
  series = 	 {Proceedings of Machine Learning Research},
  month = 	 Jul,
  publisher =    {PMLR}
}

@inproceedings{NEURIPS2024_8a23a95e,
 author = {Chen, Fan and Foster, Dylan J. and Han, Yanjun and Qian, Jian and Rakhlin, Alexander and Xu, Yunbei},
 booktitle = {Advances in Neural Information Processing Systems},
 editor = {A. Globerson and L. Mackey and D. Belgrave and A. Fan and U. Paquet and J. Tomczak and C. Zhang},
 pages = {75585--75641},
 publisher = {Curran Associates, Inc.},
 title = {Assouad, Fano, and Le Cam with Interaction: A Unifying Lower Bound Framework and Characterization for Bandit Learnability},
 volume = {37},
 year = {2024}
}

@book{ebanks1998characterization,
  title={Characterization of information measures},
  author={Ebanks, Bruce and Sahoo, Prasanna K and Sander, Wolfgang},
  year={1998},
  publisher={World Scientific}
}

@article{shannon1948mathematical,
  title={A mathematical theory of communication},
  author={Shannon, Claude Elwood},
  journal={The Bell system technical journal},
  volume={27},
  number={3},
  pages={379--423},
  year={1948},
  publisher={Nokia Bell Labs}
}

@article{csiszar2008axiomatic,
  title={Axiomatic characterizations of information measures},
  author={Csisz{\'a}r, Imre},
  journal={Entropy},
  volume={10},
  number={3},
  pages={261--273},
  year={2008},
  publisher={Molecular Diversity Preservation International}
}

@article{HobsonArthur1969Anto,
author = {Hobson, Arthur},
journal = {Journal of statistical physics},
number = {3},
pages = {383-391},
title = {A new theorem of information theory},
volume = {1},
year = {1969},
}

@article{KannappanPl1973Msof,
author = {Kannappan, Pl and Ng, C. T.},
issn = {0002-9939},
journal = {Proceedings of the American Mathematical Society},
number = {2},
pages = {303-310},
publisher = {American Mathematical Society},
title = {Measurable solutions of functional equations related to information theory},
volume = {38},
year = {1973},
}

@article{kannappan1974functional,
  title={On functional equations connected with directed divergence, inaccuracy and generalized directed divergence},
  author={Kannappan, Palaniappan and Ng, Che},
  journal={Pacific Journal of Mathematics},
  volume={54},
  number={1},
  pages={157--167},
  year={1974},
  publisher={Mathematical Sciences Publishers}
  }

\end{document}